\documentclass[10pt,conference,letterpaper]{IEEEtran}
\usepackage{times,amsmath,epsfig}
\usepackage{algorithm}
\usepackage[noend]{algorithmic}
\usepackage{subfigure}
\usepackage{setspace}

\newtheorem{example}{Example}
\newtheorem{definition}{Definition}
\newtheorem{theorem}{Theorem}
\newtheorem{lemma}{Lemma}

\newtheorem{proof}{Proof}

\DeclareMathOperator*{\argmax}{argmax}

\title{Discovery of Paradigm Dependencies}
\author{%
{Jizhou Sun{\small $~^{1}$}, Jianzhong Li{\small $~^{2}$}, Hong Gao{\small $~^{3}$} }%
\vspace{1.6mm}\\
\fontsize{10}{10}\selectfont\itshape
School of Computer Science and Technology, Harbin Institute of Technology\\
92 West Dazhi Street, Nan Gang District, Harbin, China\\
\fontsize{9}{9}\selectfont\ttfamily\upshape
$^{1}$\,sjzh@hit.edu.cn, $^{2}$\,lijzh@hit.edu.cn, $^{3}$\,honggao@hit.edu.cn
}
\begin{document}
\maketitle
\begin{abstract}
Missing and incorrect values often cause serious consequences.
To deal with these data quality problems, a class of common employed tools are dependency rules, such as \emph{Functional Dependencies} (FDs), \emph{Conditional Functional Dependencies} (CFDs) and \emph{Edition Rules} (ERs), etc.
The stronger expressing ability a dependency has, data with the better quality can be obtained.
To the best of our knowledge, all previous dependencies treat each attribute value as a non-splittable whole.
Actually however, in many applications, part of a value may contains meaningful information, indicating that more powerful dependency rules to handle data quality problems are possible.

In this paper, we consider of discovering such type of dependencies in which the left hand side is part of a regular-expression-like \emph{paradigm}, named \emph{Paradigm Dependencies} (PDs).
PDs tell that if a string matches the paradigm, element at the specified position can decides a certain other attribute's value.
We propose a framework in which strings with similar coding rules and different lengths are clustered together and aligned vertically, from which PDs can be discovered directly.
The aligning problem is the key component of this framework and is proved in \emph{NP-Complete}.
A greedy algorithm is introduced in which the clustering and aligning tasks can be accomplished simultaneously.
Because of the greedy algorithm's high time complexity, several pruning strategies are proposed to reduce the running time.
In the experimental study, three real datasets as well as several synthetical datasets are employed to verify our methods' effectiveness and efficiency.
\end{abstract}

% NOTE keywords are not used for conference papers so do not populate them
% \begin{keywords}
% keyword-1, keyword-2, keyword-3
% \end{keywords}
%
\section{Introduction}
Statistics show that dirty data is becoming more inevitable and widespread\cite{Miller, Swartz}, which often causes serious consequences\cite{Ray, Otto} and is expensive to clean.
In recent years, the database communities have extensively investigated the problem of dealing with dirty data.
Inconsistency and incompleteness are two important aspects of dirty data.
A database is inconsistent if it violates some data quality rules such as functional dependencies\cite{Codd}, conditional functional dependencies\cite{Fan,Bohannon}, extended conditional functional dependencies\cite{Bravo} and editing rules\cite{Fan2}, fixing rules\cite{Wang}, etc.
These rules are also helpful in imputing missing values in an incomplete database.

All these rules exploit relationships between entire attributes, just as defined in relational database where attributes are non-splittable.
In many actual applications however, part of an attribute value (especially one with string-type) contains information useful in dealing with incomplete and inconsistent data:
manufacturer of a product may name it by the specifications, an ISBN number contains information of the press, DOI number of a paper indicates its publishing information such as the organization, the volume and issue number, etc.

In this paper, we consider of discovering such type of dependencies (\emph{Paradigm Dependencies}, PDs) from existing datasets.
Take a motivating example, in a online shopping website, such as \emph{eBay}, \emph{Amazon} and \emph{Rakuten}, etc., kinds of products with their specifications are listed in the demonstrating pages.
However, the specification data may contain errors or even be incomplete.
Error information may mislead the customers to buy goods they dislike.
Commodities with incomplete information may be neglected in searching.
For instance, a customer is looking for computers with memory size equals $16GB$, a product meeting the requirement will become invisible to the potential buyer if its memory size value is unknown.
To avoid these undesirable things from happening, PDs say that part of the id (or type identifier, serial number etc.) of a product can help finding its correct specification values.
We employee a real world example to illustrate the feasibility:
\begin{example}
\label{exmp1}
\emph{SL410}, \emph{T520i} and \emph{T560} are three notebook types of \emph{Thinkpad}, as shown in table \ref{tab1}.
They are not named casually, but have followed some laws: the starting letter stands for its serial, the following numeric digit for \emph{Screen Size} (\emph{4} for \emph{14 inch} and \emph{5} for \emph{15 inch}), and the second digits of \emph{T520i} and \emph{T560} show that they are firstly sold in \emph{2012} and \emph{2016} respectively.

If the \emph{Screen Size} value of \emph{T560} is unfortunately lost, according to the rule 'The first numeric digit decides the screen size value' along with \emph{T520i}'s information, it is easy to derive the missing value, while in the traditional dependencies it is not the case.
\end{example}

\begin{table}[!hbp]
\centering
\begin{tabular}{|c|c|c|c|c|}
\hline
\underline{Type} & Aligned Type & Year & Screen Size & ...\\
\hline
SL410 & SL410\_ & 2010 & 14 inch & ... \\
T520i & T\_520i & 2012 & 15 inch & ... \\
T560 & T\_560\_ & 2016 & 15 inch & ... \\
\hline
New S1 2016& New S1 2016 & 2016 & 12 inch & ... \\
New S1 2017& New S1 2017 & 2017 & 13 inch & ... \\
\hline
... & ... & ... & ... & ... \\
\hline
\end{tabular}
\caption{Several notebook types of \emph{Thinkpad}}
\label{tab1}
\end{table}
%\vspace{-2em}

From real world applications, several observations can be found in the naming laws:
\begin{itemize}
\item The string values are of different lengths, and digits with the same meaning between different string may occur in different positions.
\item The meaningful digits are often ordered, for example, the digits standing for screen size always appear before those for year.
\item Digits with similar meanings are also similar to each other.
In the running example, the first numeric digit stands for screen size. \emph{4} and \emph{5} are not far from each other in the alphabet.
\item Not all tuples obey the same naming law, e.g., \emph{New S1 2016} and \emph{New S1 2017} are named in a different way from the other three.
\end{itemize}

It can be found that there are indeed available information in part of a string-type attribute.
Meaningful naming laws exist in many areas in different ways, making it hard to obtain manually, and discovering rules automatically is of great significance.

A intuitive solution is \textbf{clustering} and \textbf{aligning} the strings first, then \textbf{detecting} dependencies in each clustered group.
By clustering, strings under similar naming laws should be clustered together, and by aligning, characters with similar meaning should be aligned to the same column.
In table \ref{tab1}, for instance, the first three types may be clustered together and aligned by inserting $null$ values (presented by underlines in column \emph{Aligned Type}).
Dependency relationships between these aligned positions and other attributes can be found and expressed by star-free regular expressions in the detecting step, where several measures are needed, such as \emph{support} and \emph{confidence}, etc., similar to those in \cite{Fei}.

Contributions of this paper include:
\begin{itemize}
\item A framework of finding paradigm dependencies is introduced.
\item The problem of string aligning is analysed and proven in NP-Complete.
\item A greedy algorithm is proposed and optimized, by which the aligning and clustering tasks are combined seamlessly.
\item Experiments are conducted on both real and synthetical datasets to verify our methods' performance.
\end{itemize}

In the rest of this paper, related work is summarized in section \ref{relatedwork}.
In section \ref{clusteringalgnment}, the problems of clustering and aligning, are studied in detail.
In section \ref{finding} we discussed how to find dependencies using the clustered and aligned attributes.
In section \ref{experiments} are our experimental results and the paper is concluded in \ref{conclusion}.

\section{Related Work}
\label{relatedwork}
\subsection{Data Quality Rules}
Arenas et al. proposed the concept of consistent query answer with respect to violation of traditional FDs, which identify permanent answers over all possible repairs of a dirty database\cite{Arenas}.
In \cite{Bohannon}, Bohannon et al. introduced CFDs by introducing embedded values into FDs, which have stronger expressing ability.
Matching dependencies \cite{Fan3} were introduced to expressing an expert's knowledge in improving data quality.
MDs declare that if two tuples are similar enough on some attributes, they should equal on a certain other attribute.
Fan et al. proposed editing rules in \cite{Fan4}, which can find certain fixes based on some already known clean information.
Given some certain information, editing rules can tell which attributes to fix and how to update them.
In \cite{Jiannan}, fixing rules were proposed.
A fixing rule contains an evidence pattern, a set of negative patterns and a fact value.
When a tuple fits some negative patterns, the corresponding fixing rule can capture the wrong values and knows how to fix them.
For dirty timestamps, temporal constraints were introduced in \cite{Shaoxu}, which declare the distance restrictions between timestamps.
Guided by the constraints, the problem of repairing inconsistent timestamps was studied.
In \cite{Zeyu}, regular expressions (REs) can be used to repair dirty data, making them obeying the given REs with minimal revision cost.
\subsection{Rules Discovery}
There are numerous of algorithms proposed for discover FDs, including Tane\cite{Huhtala}, DepMiner\cite{Lopes}, FastFDs\cite{Wyss}.
Both Tane and DepMiner search the attribute lattice in a levelwise manner, and can directly obtain minimal FD cover.
On the other hand, FastFDs search FDs in a greedy heuristic and depth-first manner, and may lead to non-minimal FDs, a further checking progress is required.

In \cite{Fei}, the problem of discovering CFDs was studied, algorithms for discovering minimal CFDs and dirty values in a data instance were presented.
Fan et al. proposed algorithms discovering CFDs analog to those for FDs, i.e., CFDMiner for discovering constant CFDs, and Ctane and FastCFD for ordinary ones in \cite{Fan5}.

Besides discovery of FDs and CFDs, there have been researches focus on discovering other rules.
Discovery of \emph{Denial Constraints} (DCs) can be found in \cite{Xu}, an efficient instance-driven discovery algorithm (extended from FastFD) was developed based on a set of inference rules.
For temporal rules used in web data cleaning, \cite{Ziawasch} used machine learning methods such as association measures and outlier detection to handle the discovery problem, which is robust to noise.
\section{Clustering and Aligning}
\label{clusteringalgnment}
In this section, we introduce the clustering and aligning problems, and then analyze the difficulties in solving them.
Table \ref{tab_symbols} summarizes the most important symbols used in the rest sections of the paper.

\begin{table}[!hbp]
\centering
\begin{tabular}{|c|c|}
\hline
Notation & Description\\
\hline
$d(\cdot)$ & Distance Function \\
$D(C)$ & Diameter of a set of elements \\
$length(P)$ & Length of strings in paradigm $P$ \\
$card(P)$ & Number of strings in $P$ \\
$P[i]$ & The set of characters at the $i$th column of $P$ \\
$size(P)$ & Sum of $D(P[i])$ for all location $i$s \\
$P_1\uplus P_2$ & A new paradigm created by \\
 & aligning (or merging) $P_1$ and $P_2$ \\
$lb(P_1,P_2)$ & Lower bound of $size(P_1\uplus P_2)$ \\
$ub(P_1,P_2)$ & Upper bound of $size(P_1\uplus P_2)$ \\
$I(P_1,P_2)$ & Bound Interval of $size(P_1\uplus P_2)$ \\
$\mathcal{CR}$ & The set of critical intervals \\
$\mathcal{P(CR)}$ & Paradigms involved in $\mathcal{CR}$ \\
\hline
\end{tabular}
\caption{Several Important Frequently Used Symbols}
\label{tab_symbols}
\end{table}

\subsection{Clustering}
In a good clustering method, two properties are desired:
\begin{itemize}
\item Strings with similar naming law should be clustered together.
\item Strings with different name laws should be clustered into different groups.
\end{itemize}
For the first property, if strings with similar naming law are separated, the number of strings in a group will be reduced, which means low supports.
On the other hand, if strings with different naming laws are grouped together, existing dependencies would be discarded due to underestimated confidences.
All these undesired conditions will prevent true dependencies from being discovered.

In a classical clustering, one or more parameters are often required to control the number or sizes of clusters.
Unfortunately, in our situation, it is hard to give such a parameter.
In addition, simply grouping strings into different sets are insufficient.
For example, if dependencies $\varphi_1$ and $\varphi_2$ can be discovered from groups $G_1$ and $G_2$ respectively, it is possible that by merging $G_1$ and $G_2$ into $G_3$, a new dependency $\varphi_3$ can be detected.

In this paper, we consider of an adaptive method similar to hierarchical clustering \cite{Guha,Zhang,George}.
In our proposed framework, clustering parameters are not required and groups like $G_1,G_2$ and $G_3$ are all maintained.
\subsection{Aligning}
Let $\mathcal{S}$ be a set of strings over character set $\mathcal{C}$.
Given a string $s$, $length(s)$ is the number of characters (or elements) in $s$, and the $i$th character in $s$ is denoted by $s[i]$, where $1\le i\le length(s)$.

Before the problem definition, we introduce several notations.
A distance function over $\mathcal{C}$ is:
$$d:\mathcal{C}\times\mathcal{C}\rightarrow[0,+\infty]$$
In the rest of this paper, $d$ is metric by default, if not specified.

\begin{definition}[Diameter of Character Sets]
For a set $C\subseteq\mathcal{C}\cup\{null\}$ of characters, the diameter of $C$ is the largest distance between every pair of elements in $C$:
$$D(C)=\max\limits_{c_1, c_2\in C}d(c_1,c_2)$$
\end{definition}
Over '$\cup$', the diameter satisfies monotonicity and triangle inequality:
\begin{lemma}
For two arbitrary characters $C$ and $C'$, if $C\subseteq C'$ $D(C)\le D(C')$ always holds. \emph{(Proof omitted.)}
\end{lemma}

\begin{lemma}
For arbitrary three character sets $C_1$, $C_2$ and $C_3$, the following equations holds:
\begin{equation}
\label{eq_set_tria}
D(C_1\cup C_2 \cup C_3)\le D(C_1\cup C_2) + D(C_2\cup C_3)
\end{equation}
\end{lemma}
\begin{proof}
Let $c_1$ and $c_2$ be the pair of farthest characters in $C_1\cup C_2\cup C_3$.
We prove the inequality in all possible cases:
\begin{itemize}
\item If $c_1,c_2\in C_1\cup C_2$, then $D(C_1\cup C_2 \cup C_3) = D(C_1\cup C_2) \le D(C_1\cup C_2) + D(C_2\cup C_3)$.
\item If $c_1,c_2\in C_2\cup C_3$, the proof is similar.
\item In the rest case, $c_1$ and $c_2$ must come from $C_1$ and $C_3$ respectively, without of loss of generality, assume $c_1\in C_1$ and $c_2\in C_3$, for any character $c'\in C_2$:
$D(C_1\cup C_2)\ge d(c_1,c')$ and $D(C_2\cup C_3)\ge d(c',c_2)$.
Meanwhile, because $d$ is a metric, we have $d(c_1,c')+d(c',c_2)\ge d(c_1,c_2)=D(C_1\cup C_2\cup C_3)$, indicating Eq. \ref{eq_set_tria}.
\end{itemize}
\end{proof}

\begin{definition}[Paradigm]
A paradigm $P$ is a set of strings with equal length.
The common length of strings in $P$ is also called $P$'s length, denoted as $length(P)$, and cardinality of $P$ is defined as the number of strings in $P$, denoted as $card(P)$.
The size of $P$ is defined by:
$$size(P)=\sum\limits_{i=1}^{length(P)}D(P[i])$$
where $P[i]$ is the set of characters in the $i$th column:$P[i]=\{s[i]|s\in P\}$.
\end{definition}

Now we are ready to define the aligning problem formally.
\begin{definition}[String Aligning Problem (SAP)]
Given a set of strings $\mathcal{S}=\{s_1,s_2,...,s_N\}$, a distance function $d$ and a threshold $t$, find out a paradigm $P=\{s_1'.s_2',...,s_N'\}$, where $s_i'$ is obtained from $s_i$ by inserting $null$ values, such that $size(P)\le t$.
\end{definition}

The intuition behind is that similar characters will be aligned into a single column with high probability if the distance function is well designed.
Only insertions are allowed because delete or change a character will cause loss of information, which is undesirable.
The intuition of using $D(P[i])$ to evaluate the size of a paradigm is that, smaller diameter means higher similariry degree.
With $size(P)$ minimized, similar characters would be aligned together with high possibility.
Meanwhile, number of distinct elements in $P[i]$ should not influence the size too much.
For example, if '$1$' and '$3$' is already in $P[i]$, '$2$' is natural to add into this set, and should not enlarge the size too much.
For a contrary example, $\sum\limits_{c_1,c_2\in C}d(c_1,c_2)$ does not meet the latter property.

Unfortunately, the \emph{SAP} is intractable.
\begin{theorem}
\emph{SAP} is in \emph{NP-Complete}, even if $d$ is a metric.
\end{theorem}
\begin{proof}
The lower and upper bounds of $length(P)$ is $\max\limits_{s\in\mathcal{S}}length(s)$ and $\sum\limits_{s\in\mathcal{S}}length(s)$ respectively.
To prove $SAP\in NP$, we just guess $length(P)$ and the positions to insert $null$s for each string in $\mathcal{S}$, and then verify whether the paradigm size is no more than $t$.
Obviously, this can be done by a nondeterministic algorithm in polynomial time, so $SAP\in NP$

The \emph{NP-Completeness} can be proven by reduction of the Shortest Common Super-sequence (SCS) problem\cite{Maier}.

The SCS problem is that for a set $S_{SCS}$ of strings, find out a string $s'$ with $len(s')\le k$, such that for any string $s\in S_{SCS}$, $s$ is a subsequence of $s'$.
Given such a instance of SCS, the construction of String Aligning Problem (SAP) need some tricks.
Let $S_{SAP}$ be the string set in SAP, $$S_{SAP}=S_{SCS}\cup\{s_n\}$$ where $s_n$ is a string constituted with $k$ identical characters $c_n$, which is a new character hasn't appeared in $S_{SCS}$.

The distance function $d$ is defined as follows:
\begin{equation*}
d(a,b)=
\begin{cases}
0 &\mbox{if $a = b$}\\
1 &\mbox{else if $a\ne b\wedge(a=null\vee b=null)$}\\
1 &\mbox{else if $a\ne b\wedge(a=c_n\vee b=c_n)$}\\
2 &\mbox{else}
\end{cases}
\end{equation*}

The threshold is $t = k$.
It is easy to verify that $d$ satisfies the triangular inequality, and the construction can be done in polynomial time.

Now we prove that there exists a common super-sequence $s'$ for the SCS Problem with $length(s')\le k$ if and only if there exists a paradigm $P$ for the SAP Problem with $size(P)\le t$.

\emph{If.}
Because $s_n$ is in the string set, and only insertions are allowed, so $length(P)\ge length(s_n)=k$.
There are exactly $k$ columns containing the character $c_n$, and everyone of them contains a character different from $c_n$ (i.e., $null$ or a ordinary one).
For such a single column, the diameter of the corresponding character set satisfies $D(P[i])\ge 1$ according to the definition of $d$.
Thus we have $$size(P)\ge D(P[i])\times length(P)\ge 1\times k = k$$
Along with the \emph{If} condition $size(P)\le t = k$, we have $size(P)=k$, which means $length(P)=k$ and there are no more than 2 different not-null characters at each column.
By removing the string $s_n$ from $P$, the rest strings can be directly merged into a super-sequence with length equals $k$.

\emph{Only If.}
If there exists a common super-sequence $s'$ for the SCS Problem with $length(s')\le k$, each string $s_i$ in the string set can be transformed into $s'$ by inserting characters into some places.
By inserting $null$s instead of specified characters, into these places, a new string $s_i'$ with length equals $k$ can be obtained.
The new strings along with $s_n$ formed the solution of the SAP Problem.
It is obvious that $length(P)=k$.
At each column, there are no two characters with distance larger than 1, so the diameter at each column is no more than 1 and $size(P)\le length(P)\times1=k$.
\end{proof}

Due to its intractability, we consider of greedy solutions.
The overall framework and algorithms are introduced in the rest of this section.

\subsection{Framework}
\label{framework}
Although the aligning problem is intractable when the number of strings is arbitrary, there exist determined algorithms giving the optimal resolution for a constant number of strings, by \emph{Dynamic Programming} (DP).
Actually, the problem of calculating \emph{Edit Distance} is a specialization of the aligning problem with two input strings, by specifying the distance between different characters as a constant value.
The algorithm is almost all the same, except for some difference in the recursive equation.
In our setting, the recursive equation becomes:
\begin{equation*}
size_{i,j}=min
\begin{cases}
size_{i-1,j}+d(s_1[i], null)\\
size_{i,j-1}+d(s_2[j], null)\\
size_{i-1,j-1}+d(s_1[i],s_2[j]))
\end{cases}
\end{equation*}
In calculating \emph{Edit Distance}, the distance function $d$ is replaced with the constant $1$.
In the rest of this paper, in situations of two strings, we use the word '\emph{merge}' instead of '\emph{align}' timely for sake of easy understanding.

Due to tractability of the two-string case and intractability of the general case, a straightforward idea is merging the most similar strings at each step ,i.e., in a greedy manner.
The intuition behind is that strings with similar naming law should be merged as early as possible.
By doing this, different characters with similar meanings will be aligned together with high possibilities.
For example, '$ab1c$' and '$ab9c$' are obviously obeying the same naming law.
If they are merged firstly, the three identical characters '$a$', '$b$' and '$c$' can help us align the different characters '$1$' and '$9$' together.

We say that the aligning result paradigm is \emph{merged} from the input strings.
Generally, merged paradigms can be further merged with other strings or paradigms in similar way.
Without loss of generality, a single string can be seen as a initial paradigm.
Then we can denote by $P_1\uplus P_2$ the merged paradigm from paradigms $P_1$ and $P_2$.
We call $P_1\uplus P_2$ the super-paradigm of $P_1$ (or $P_2$), and $P_1$ ($P_2$, resp.) is a sub-paradigm of $P_1\uplus P_2$.

There are two optional greedy algorithms:

In the first option, two strings in $\mathcal{S}$ are merged into a paradigm $P$, with $size(P)$ minimized.
Next, a string in $S$ is selected and merged into $P$, which has never been selected before and with $size(P)$ still minimized.
This process continues until all strings in $S$ are merged into $P$.
Because a single paradigm is being enlarged in the running time, we call it \emph{Single Merge}, as shown in algorithm \ref{alg_single_merge}.
\begin{algorithm}[h]
    \caption{Single Merge}
    \begin{algorithmic}[1]
    \label{alg_single_merge}
    \REQUIRE A set $\mathcal{S}$ of strings over charset $\mathcal{C}$ and a metric $d$ defined on $\mathcal{C}$.
    \ENSURE An aligning of $\mathcal{S}$ with size as small as possible.
    \STATE {$P \leftarrow \{s_1\}\uplus\{s_2\}$, where $s_1$ and $s_2$ are selected from $\mathcal{S}$ with $size(\{s_1\}\uplus\{s_2\})$ minimized.};
    \STATE {$\mathcal{S}\leftarrow\mathcal{S}-\{s_1,s_2\}.$}
    \WHILE {$|\mathcal{S}|>0$}
    \STATE {Select a string $s$ in $\mathcal{S}$ with $size(P\uplus\{s\})$ minimized.}
    \STATE {$P\leftarrow P\uplus\{s\}$.}
    \STATE {$\mathcal{S}\leftarrow\mathcal{S}-\{s\}.$}
    \ENDWHILE
    \RETURN $P$;
    \end{algorithmic}
\end{algorithm}

In the alternative option, each string in $\mathcal{S}$ is initialized as a single-string paradigm.
All initial paradigms are added into the set $\mathcal{P}$.
Then all the paradigms are merged iteratively in a pair-wise manner: at each step, a pair of paradigms $P_1$ and $P_2$ are selected and merged into a new one, with $size(P_1\uplus P_2)$ minimized.
After each merging operation, $P_1$ and $P_2$ are removed from $\mathcal{P}$ and $P_1\uplus P_2$ is added.
This process continues until a single paradigm is obtained in $\mathcal{P}$.
Because a pair of paradigms are merged at each step, we call it \emph{Pairwise Merge}, as shown in algorithm \ref{alg_pairwise_merge}.
\begin{algorithm}[h]
    \caption{Pairwise Merge}
    \begin{algorithmic}[1]
    \label{alg_pairwise_merge}
    \REQUIRE A set $\mathcal{S}$ of strings over charset $\mathcal{C}$ and a metric $d$ defined on $\mathcal{C}$.
    \ENSURE An aligning of $\mathcal{S}$ with size as small as possible.
    \STATE $\mathcal{P} \leftarrow \emptyset$;
    \FOR {each string $s\in\mathcal{S}$}
    \STATE Add paradigm $\{s\}$ into $\mathcal{P}$
    \ENDFOR
    \WHILE {$|\mathcal{P}|>1$}
    \STATE Find out two paradigms $P_1$ and $P_2$ in $\mathcal{P}$, with $size(P_1\uplus P_2)$ minimized.
    \STATE Add $P_1\uplus P_2$ into $\mathcal{P}$.
    \STATE Remove $P_1$ and $P_2$ from $\mathcal{P}$.
    \ENDWHILE
    \RETURN the single paradigm in $\mathcal{P}$;
    \end{algorithmic}
\end{algorithm}

The most time-consuming operation is '$\uplus$', which has a time complexity of $O(n^2\times M^2)$ where $n$ is the length of inputting paradigms, and $M$ is the charset size.
It can be shown that both \emph{Single Merge} and \emph{Pairwise Merge} require $O(N^2)$ times of merging, where $N$ is the number of strings $|\mathcal{S}|$.

\emph{Pairwise Merge} is more preferred in this paper, for it naturally fits a \emph{Hierarchical Clustering} algorithm's requirement without extra workloads:
Each paradigm can be seen as a group of similar strings, if two paradigms $P_1$ and $P_2$ are similar enough, they would be merged into a larger one $P$.

Another benefit is that, no clustering parameter is required in the pair-wise framework.
By maintaining $P_1$, $P_2$ and $P$ in a binary tree structure, discovery dependency becomes straightforward, which will be discussed in section \ref{finding}.

Due to the high algorithm complexity, we consider several pruning strategies to improve the efficiency, with necessary theoretical supports.

As discussed before, the most time-consuming operation in the pair-wise algorithm is merging of paradigms.
More over, merging is carried out for $O(N^2)$ times, even though only $N-1$ of them are actually performed to generate larger paradigms (in Line 6 of algorithm \ref{alg_pairwise_merge}) and all of the rest are just for finding out the pair with smallest size (Line 5).
So it is possible to reduce the number of useless merging work, and we propose two pruning techniques.

\subsection{Bound Based Pruning}
\label{bound_pruning}
Instead of calculating the exact size of merging each pair of paradigms, a size interval $[lb,ub]$ may requires much less computation.
We denote the lower and upper bound of $size(P_1\uplus P_2)$ by $lb(P_1,P_2)$ and $ub(P_1,P_2)$ respectively.
The corresponding interval is denoted by $I(P_1,P_2)$.

The basic idea is that when finding the pair of paradigms with the lowest merging size, we maintain a interval for each candidate pair.
If by $ub_{min}$ we denote the lowest upper bound of these intervals, all candidate pairs $P_1, P_2$ with no higher lower bounds than $ub_{min}$ can be safely pruned in the current iteration.
For all of the rest candidates, we call them critical intervals, denoted by $\mathcal{CR}$, tightening their bounds can help pruning more candidates.
The tightening process continues until a single candidate left in $\mathcal{CR}$, who is right the pair to be merged.
We call this the refining process, and illustrate it by example.
\begin{example}
\label{exmp2}
For $4$ candidate intervals $I(P_1,P_2)$, $I(P_3,P_4)$, $I(P_2,P_3)$ and $I(P_1,P_3)$ in Fig. \ref{fig1a}.
$I(P_1,P_2)$ has the lowest upper bound $ub_{min}$, $I(P_3,P_4)$ and $I(P_2,P_3)$'s lower bounds are lower than $ub_{min}$, it is hard to tell which pair of paradigms is of the lowest size corresponding to them.
On the other hand, $I(P_1,P_3)$'s lower bound is greater than $ub_{min}$, the candidate pair's merging size must not be the lowest, and can be safely pruned, presented by dotted line.
Thus the critical intervals $\mathcal{CR}$ contains $I(P_1,P_2)$, $I(P_3,P_4)$ and $I(P_2,P_3)$.

By refining intervals in $\mathcal{CR}$, as shown in Fig. \ref{fig1b}, $ub_{min}$ becomes lower and $I(P_3,P_4)$ and $I(P_2,P_3)$'s lower bounds becomes higher.
$I(P_2,P_3)$'s lower bound becomes higher than $ub_{min}$ and can be removed from $\mathcal{CR}$ and pruned.
By now only two intervals remain in CR, i.e., $I(P_1,P_2)$ and $I(P_3,P_4)$.
By a further refining, $ub_{min}$ becomes even lower and $I(P_3,P_4)$ is pruned, $I(P_1,P_2)$ is identified to be the pair with minimal merging size, as shown in Fig. \ref{fig1c}.
\end{example}

\begin{figure}
  \centering
  \subfigure[]{
    \label{fig1a} %% label for first subfigure
    \includegraphics[width=0.4\textwidth]{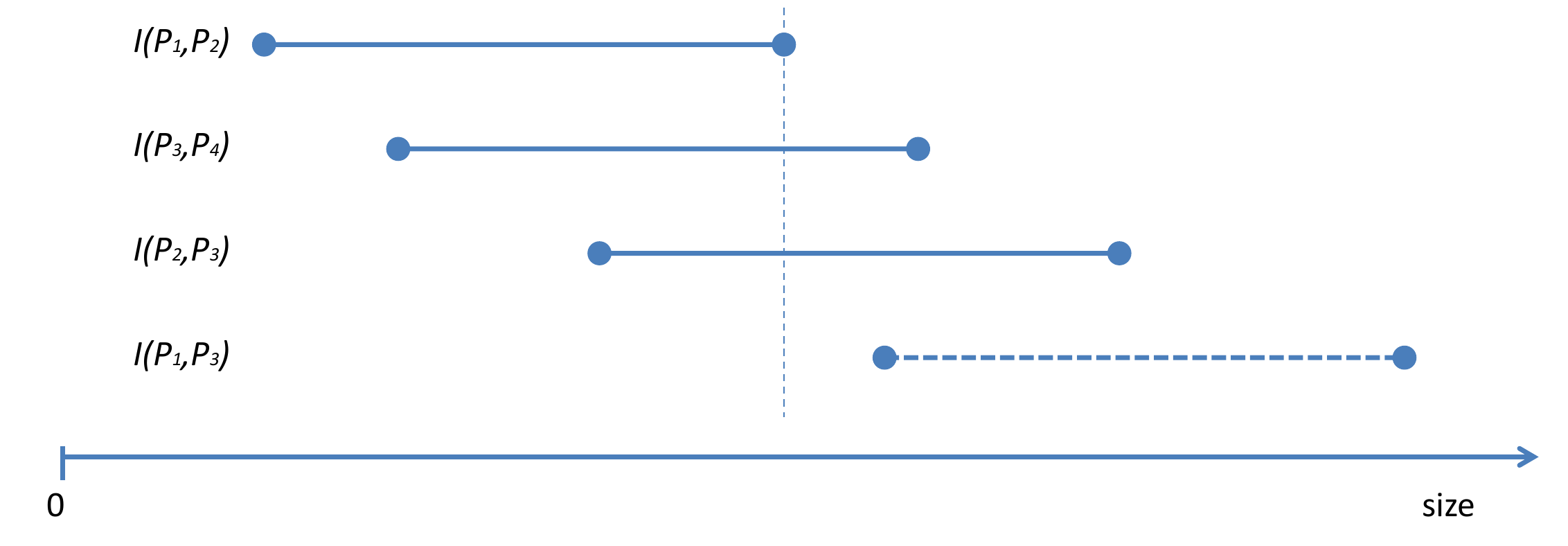}}
  \subfigure[]{
    \label{fig1b} %% label for second subfigure
    \includegraphics[width=0.4\textwidth]{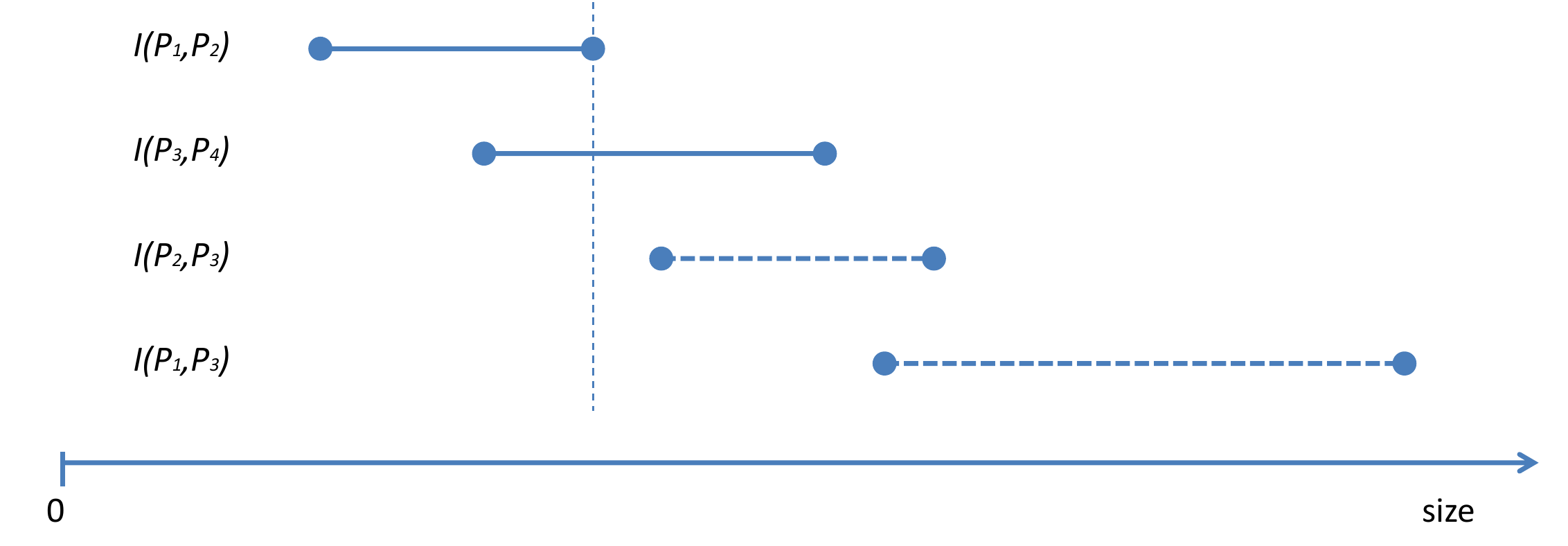}}
  \subfigure[]{
    \label{fig1c} %% label for third subfigure
    \includegraphics[width=0.4\textwidth]{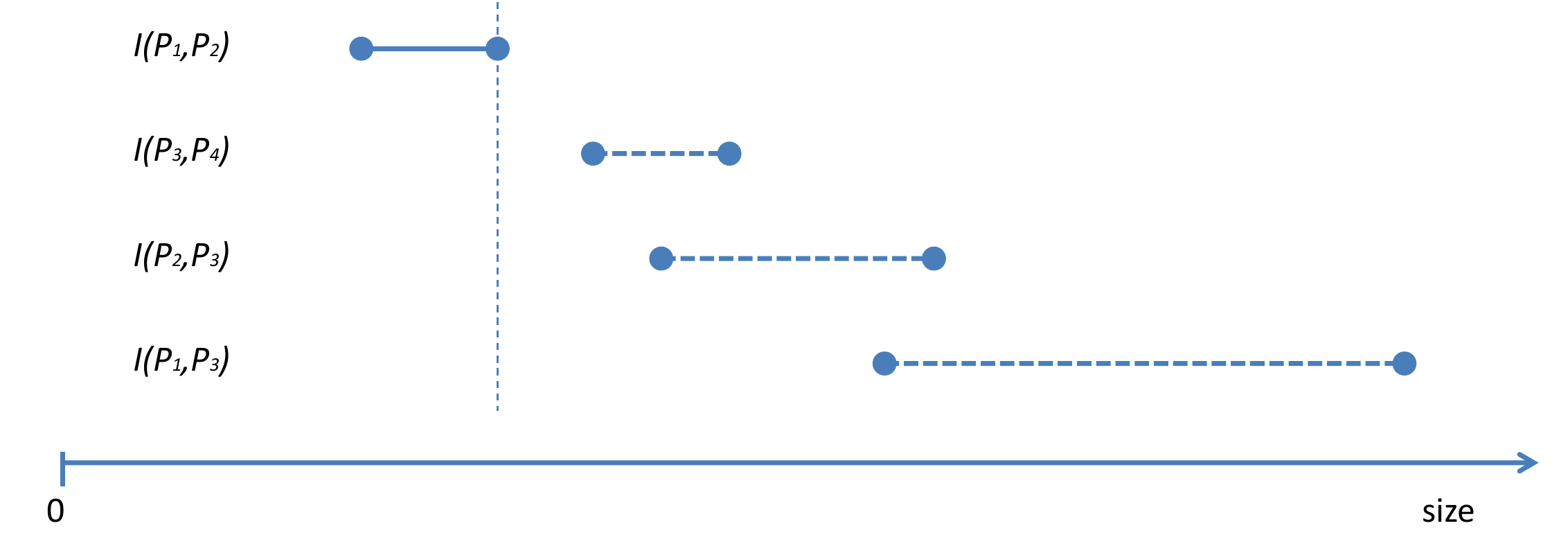}}
  \caption{An example of interval refining}
\label{fig1}
\end{figure}

Before describing the pruning based algorithm, we introduce several necessary inequalities, and discuss how to evaluating the bounds using to them.

\noindent{\textbf{Monotonicity}}

Size of a paradigm is monotonic.
\begin{lemma}
Merging a paradigm $P'$ into $P$ doesn't decrease its size:
\begin{equation}
\label{eq_par_mono1}
size(P\uplus P')\ge size(P)
\end{equation}
\end{lemma}

More over, The size of merging paradigm $P$ with a super-paradigm is no less than that of merging with a sub-paradigm.
\begin{lemma}
\begin{equation}
\label{eq_par_mono2}
size((P_1\uplus P_2)\uplus P_3)\ge size(P_1\uplus P_3)
\end{equation}
\end{lemma}

These two lemmas are obvious and straightforward to prove, thus are omitted here.

\noindent{\textbf{Triangular Inequality}}
\begin{theorem}[Triangular Inequality]
\label{thm2}
The paradigm size satisfies the triangular inequalities under merging operations if $d$ is a metric, that is, for three arbitrary paradigms $P_1$, $P_2$ and $P_3$, the following equations always hold:
\begin{equation}
\label{eq_para_tria1}
size((P_1\uplus P_2)\uplus P_3) \le size(P_1\uplus P_3) + size(P_1\uplus P_2)
\end{equation}

\begin{equation}
\label{eq_para_tria2}
size(P_1\uplus P_2) + size(P_2\uplus P_3)\ge size(P_1\uplus P_3)
\end{equation}

\begin{equation}
\label{eq_para_tria3}
|size(P_1\uplus P_2) - size(P_2\uplus P_3)|\le size(P_1\uplus P_3)
\end{equation}
\end{theorem}

Eq. \ref{eq_para_tria1} has a not so \emph{triangle} form compared with Eq. \ref{eq_para_tria2}, so we call it the \emph{Pseudo Triangular Inequality}.
\begin{proof}
We introduce a new paradigm $(P_1\uplus P_2)\uplus_{P_1} P_3$.
By the subscribe '$P_1$' it means that when aligning $P_3$ with $P_1\uplus P_2$, the relative positions are the same as those in $P_1\uplus P_3$.
Obviously, $size((P_1\uplus P_2)\uplus P_3)\le size((P_1\uplus P_2)\uplus_{P_1} P_3)$.
It becomes sufficient to prove $$size((P_1\uplus P_2)\uplus_{P_1} P_3) \le size(P_1\uplus P_3) + size(P_1\uplus P_2)$$

In $(P_1\uplus P_2)\uplus_{P_1} P_3$, for a random position $pos$, we denote the character sets at $pos$ coming from $P_1$, $P_2$ and $P_3$ by $C_1$, $C_2$ and $C_3$ respectively for ease.
Because the only change of $P_1\uplus P_2$ is inserting a set of $null$s (whose size equals $0$), $size({P_1\uplus P_2})$ remains unchanged.
It becomes sufficient to prove $D(C_1\cup C_2\cup C_3)\le D(C_1\cup C_3)+D(C_1\cup C_2)$, which directly holds according to Eq. \ref{eq_set_tria}.

Eq. \ref{eq_para_tria2} can be directly proved using Eq. \ref{eq_para_tria1} and Eq. \ref{eq_par_mono2}, and Eq. \ref{eq_para_tria3} can be inferred from Eq. \ref{eq_para_tria2} by element substitutions.
\end{proof}

In our framework, $size(P_1\uplus P_2)$ is likely to be unavailable, instead, the bounds $lb(P_1, P_2)$ and $ub(P_1, P_2)$ are maintained.
Thus, the previous inequalities should be adjust before using, for example, Eq. \ref{eq_para_tria1} becomes $lb((P_1\uplus P_2),P_3) \le ub(P_1, P_3) + ub(P_1, P_2)$.
In the rest of this paper, we refer the inequalities by their adjusted versions, if not specified explicitly.

Eq. \ref{eq_par_mono1}, \ref{eq_par_mono2} and \ref{eq_para_tria1} provide bounds when a new merged paradigm is added into $\mathcal{P}$.
If $P_1$ and $P_2$ are identified as the pair with minimal size, then after adding $P_1\uplus P_2$ into $\mathcal{P}$, $lb(P_1\uplus P_2, P)$ and $ub(P_1\uplus P_2, P')$ should be initialized for each $P'$ in $\mathcal{P}$.

According to Eq. \ref{eq_par_mono1}, $size(P_1\uplus P_2)$ is a candidate value of $lb(P_1\uplus P_2, P')$.
With Eq. \ref{eq_par_mono2}, $size(P_1\uplus P_3)$ (similarly, as well as $size(P_2\uplus P_3)$) can be another candidate value of $lb(P_1\uplus P_2, P')$.
For $ub(P_1\uplus P_2, P')$, Eq. \ref{eq_para_tria1} asserts that $size(P_1\uplus P_3) + size(P_1\uplus P_2)$ and $size(P_2\uplus P_3) + size(P_1\uplus P_2)$ are available upper bounds.

On the other hand, Eq \ref{eq_para_tria2} and \ref{eq_para_tria3} play an important role when the former three are not applicable.
For instance, before the iterations begin, each paradigm contains a single string, while Eq. \ref{eq_par_mono1}-\ref{eq_para_tria1} focus on new merged paradigms.
In addition, while refining $\mathcal{CR}$, no new paradigm is created and the critical intervals require tightening.
By calculating the exact values of $size(P_1\uplus P_2)$ and $size(P_2\uplus P_3)$, interval of $size({P_1\uplus P_3})$ would be tightened.
%In this paper, we call such a paradigm as $P_2$ a \emph{pivots}.

Now we discuss how to find pairs whose sizes should be exactly evaluated to refine the critical set $\mathcal{CR}$.
By $\mathcal{P(CR)}$, we denote the set of paradigms involved in $\mathcal{CR}$.
To tighten $I(P_1,P_2)$ for all pairs $P_1,P_2\in\mathcal{P(CR)}$, we propose a \emph{single-pivot-star} refining technique:
Select a paradigm $P\in\mathcal{P(CR)}$ as the pivot, then evaluate the exact value of $size(P,P')$ for each paradigm $P'$ in $\mathcal{P(CR)}$ except $P$ it self.
By doing so, bounds for each pair of paradigms $P_1,P_2$ can be updated by:
\begin{equation*}
lb'(P_1,P_2)=
\begin{cases}
size(P_1,P_2) &\mbox{if $P\in\{P_1,P_2\}$.}\\
|size(P_1,P)-size(P,P_2)| &\mbox{else}
\end{cases}
\end{equation*}
and
\begin{equation*}
ub'(P_1,P_2)=
\begin{cases}
size(P_1,P_2) &\mbox{if $P\in\{P_1,P_2\}$.}\\
size(P_1,P)+size(P,P_2) &\mbox{else}
\end{cases}
\end{equation*}

We demonstrate this by an example
\begin{example}
In Fig. \ref{fig2}, there are $5$ paradigms, presented by nodes, among which $P_0$ is selected as the pivot.
$size(P_0\uplus P_1)$, $size(P_0\uplus P_2)$, $size(P_0\uplus P_2)$ and $size(P_0\uplus P_4)$ are evaluated exactly (equal $1.5$, $2$, $2$ and $1$ respectively), presented by solid lines and shaped like a star.
All other pairs' bounds are obtained according to Eq. \ref{eq_para_tria2} and \ref{eq_para_tria3}: $I(P_1,P_2)=I(P_1,P_3)=[0.5,3.5], I(P_1,P_4)=[0.5,2.5], I(P_2,P_3)=[0,4], I(P_2,P_4)=I(P_3,P_4)=[1,3]$.
\end{example}

The intuition of \emph{single-pivot-star} refining is that, all pairs' bounds can be evaluated exactly or by Eq. \ref{eq_para_tria2} and \ref{eq_para_tria3} directly.
Meanwhile, it carries out the minimal number (equals $|\mathcal{P(CR)}| - 1$) of necessary merging operations, because connecting all paradigms requires that much edges.

\noindent{\textbf{Pivot Selection}}

Now we show how to select a pivot wisely.
$CR$ is a set of intervals overlapping with each other.
From Fig. \ref{fig1}, it can be shown that tightening intervals help reducing $\mathcal{CR}$.
Thus we evaluate a pivot's ability of shortening intervals.
A score function on $\mathcal{P(CR)}$ is introduced:
\begin{equation}
\label{eq_score}
score(P)=\sum\limits_{I(P_1,P_2)\in\mathcal{CR}, P\in\{P_1,P_2\}}{ub(P_1,P_2)-lb(P_1,P_2)}
\end{equation}
That is, sum of widths of intervals that are related with $P$.
Then the paradigm with max score is selected as the pivot: $pivot=\mathop{\argmax_{P}score(P)}$.
In single-pivot-start refining, we can see that all intervals about $P$ shrink into exact values, the corresponding total width reduced equals $score(P)$, which should better be maximized.

\begin{figure}
\centering
\includegraphics[width=0.3\textwidth]{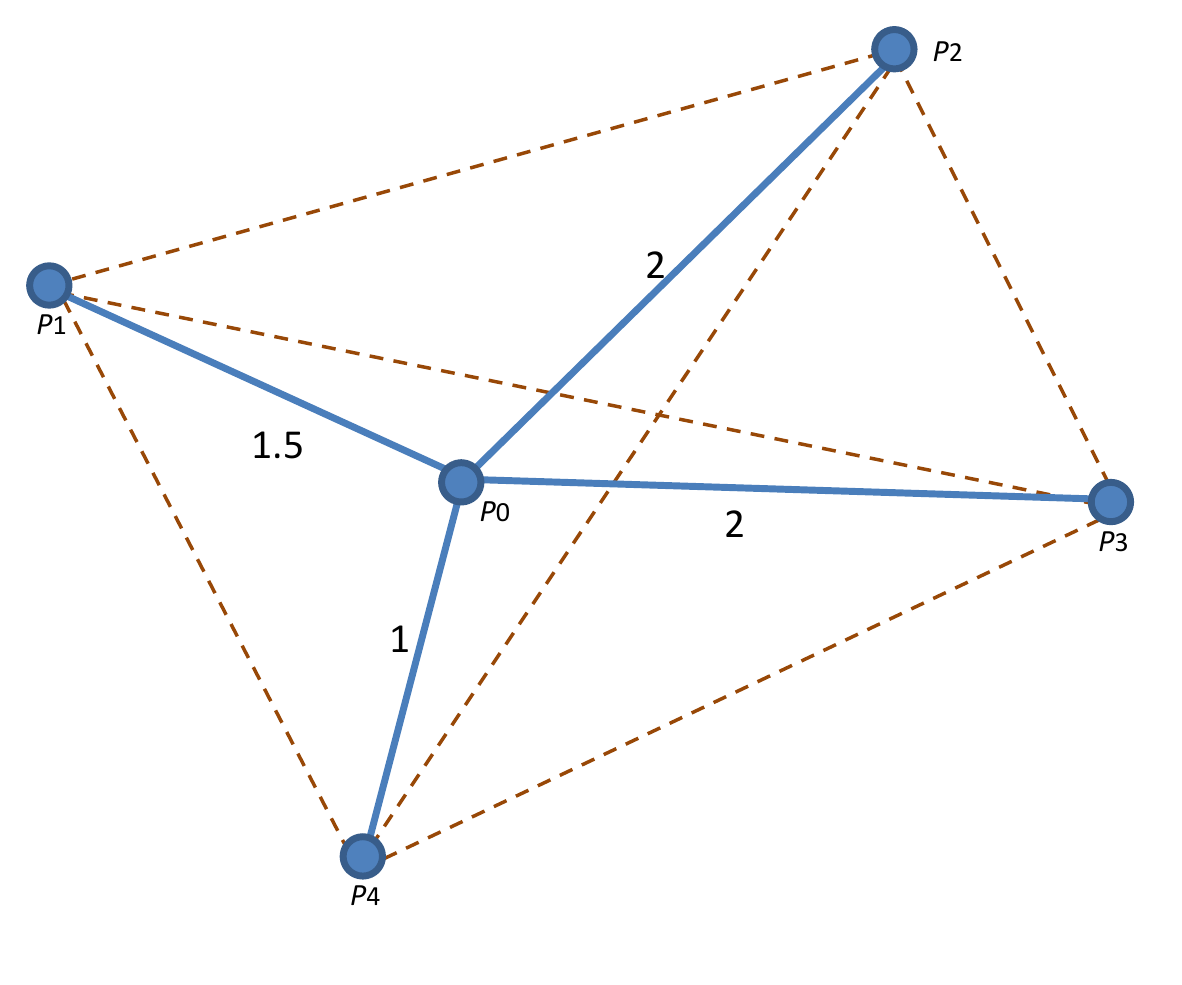}
\caption{An example of \emph{single-pivot-star} refining.}
\label{fig2}
\end{figure}

By now, pruning techniques based on lower and upper bounds have been discussed.
Next, we introduce \emph{Independency Base Pruning}, to improve the algorithm efficiency further.
\subsection{Independency Based Pruning}
\label{indepency_pruning}
When identifying $\mathcal{CR}$, up to $O(|\mathcal{P}|^2)$ intervals are identified to be critical in the worst case.
That is the case if almost all paradigm pairs are of the same size.
It is not hard to get that $|\mathcal{P(CR)}|\ge |\mathcal{CR}|$.
Meanwhile, in the refining process, up to $|\mathcal{P(CR)}| - 1$ actual merging operations are required and $(\mathcal{P(CR)}| - 1)\times(\mathcal{P(CR)}| - 2)/2$ intervals are to be refined.
As a result, $O(|\mathcal{CR}|)$ merging and $O(|\mathcal{CR}|^2)$ refining operations are performed, this will degrade the efficiency seriously.

To reduce $\mathcal{CR}$'s cardinality, it can be identified more wisely.
By $I_{um}$ we denote the interval from which $ub_{min}$ is found, e.g., $I(P_1,P_2)$ in Fig. \ref{fig1a}.
We consider of verifying those intervals who are dependent with $I_{um}$ only.
By '\emph{dependent}', it means that one interval shares a common paradigm with another.
We demonstrate this by an example:
\begin{example}
\label{exmp3}
In Fig. \ref{fig1} of Example \ref{exmp2}, $I_{um}$ is $I(P_1,P_2)$ because it holds the lowest upper bound.
$I(P_1,P_3)$ and $I(P_2,P_3)$ are dependent with $I(P_1,P_2)$, because they share a common paradigm with $I(P_1,P_2)$ (i.e., $P_1$ and $P_2$ respectively).
On the other hand, $I(P_3,P_4)$ is independent with $I(P_1,P_2)$ because they share no common paradigm.
With independency based pruning, $I(P_3,P_4)$ is not added into $CR$, and the refining terminates after a single refining process (see Fig. \ref{fig1b}).
\end{example}

Under this pruning strategy, a paradigm pair with higher size may be merged earlier than another pair with lower size.
Fortunately, this adjustment doesn't affect the final result, which can be proved theoretically.
\begin{theorem}
\label{thm_ind_prun}
If $P_1$ and $P_2$ is of the minimal merging size among all pairs involving $P_1$ or $P_2$, they will be actually merged sooner or later.
\end{theorem}
\begin{proof}
Relationship between $P_1$, $P_2$ and all other paradigms is illustrated in Fig. \ref{fig3}.
$P_1$ and $P_2$ have the minimal merging size with each other.
Note that there probably exist pairs having lower merging size than $size(P_1,P_2)$.

Now it can be shown that $P_1$ and $P_2$ will be merged eventually, no matter sooner or later.
When merging occurs between other paradigms, say $P'$ and $P''$, $P'\uplus P''$ is added in to $\mathcal{P}$, we have $size(P_1\uplus P_2)<size(P_1\uplus (P'\uplus P''))\le size(P_1,P')$, according to Eq. \ref{eq_par_mono2}.
$P_1$ and $P_2$ still have the minimal merging size with each other.
Thus, it can be concluded that when $P_1$ (or $P_2$) is merged, it must be merged with $P_2$ ($P_1$ resp.), so it makes no difference by merging them sooner or later.
\end{proof}

\begin{figure}
\centering
\includegraphics[width=0.4\textwidth]{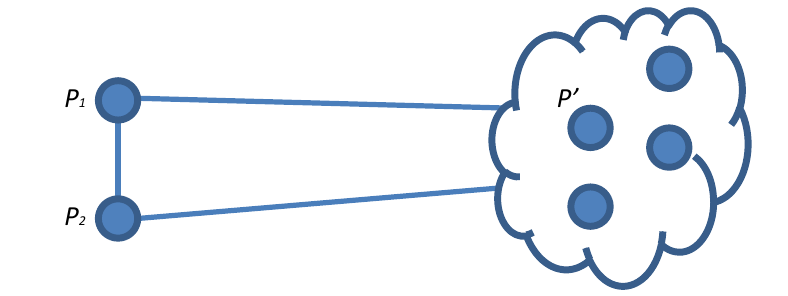}
\caption{Relationship between $P_1,P_2$ and other paradigms.}
\label{fig3}
\end{figure}

\subsection{Algorithm Flow}
Our pruning based algorithm is illustrated in Algorithm \ref{alg_pruning_merge}.
In lines 1-3, The paradigm set $\mathcal{P}$ as well as size intervals are initialized.
Then paradigms in $\mathcal{P}$ are merged in pair-wise manner iteratively, until a single one is obtained (Lines 4-15).
In each iteration, the critical interval set $\mathcal{CR}$ is identified (Lines 5-6).
Then $\mathcal{CR}$ is refined iteratively until a single interval left (Lines 7-10).
In line 8, a pivot paradigm is selected according to Eq. \ref{eq_score}.
Then intervals in $\mathcal{CR}$ are refined by calling algorithm \ref{alg_refine}, and $ub_{min}$ and $\mathcal{CR}$ are updated accordingly (Line 10).

After the refining process terminated, from the single interval in $\mathcal{CR}$, two paradigms can be identified to be merged into a new one (Line 11), the paradigm set $\mathcal{P}$ is updated (Lines 12 and 13).
For the new added paradigm $P$, the sizes' bounds of merging it with others are evaluated (Lines 14 and 15).

\begin{algorithm}[h]
    \caption{Pruning Merge}
    \begin{algorithmic}[1]
    \label{alg_pruning_merge}
    \REQUIRE A set $\mathcal{S}$ of strings over charset $\mathcal{C}$ and a metric $d$ defined on $\mathcal{C}$.
    \ENSURE An aligning of $\mathcal{S}$ with size as small as possible.
    \STATE {Initialize the paradigm set $\mathcal{P}$ with $\mathcal{S}$.}
    \FOR {each pair of paradigms $P_1,P_2\in\mathcal{S}$}
    \STATE {$I(P_1,P_2)\leftarrow [0,+\infty]$}
    \ENDFOR
    \WHILE {$|\mathcal{P}|>1$}
    \STATE {$I_{um}\leftarrow$ Interval with the lowest upper bound.}
    \STATE {Identify the critical set $\mathcal{CR}$ by $I_{um}$.}
    \WHILE {$|CR|>1$}
    \STATE {Identify the pivot paradigm $P\in\mathcal{P(CR)}$ using Eq. \ref{eq_score}.}
    \STATE {Refine($\mathcal{CR}$,$P$)}
    \STATE {Update $ub_{min}$ and $\mathcal{CR}$ accordingly.}
    \ENDWHILE
    \STATE {$P\leftarrow P_1\uplus P_2$, where $P_1,P_2\in\mathcal{P(CR)}$}
    \STATE {Add $P$ into $\mathcal{P}$.}
    \STATE {Remove $P_1$ and $P_2$ from $\mathcal{P}$}.
    \FOR {each paradigm $P'\in \mathcal{P}-\{P\}$}
    \STATE {Update $lb(P,P')$ and $ub(P,P')$ using Eq. \ref{eq_par_mono1}-\ref{eq_para_tria1}.}
    \ENDFOR
    \ENDWHILE
    \RETURN the single paradigm in $\mathcal{P}$;
    \end{algorithmic}
\end{algorithm}

Algorithm \ref{alg_refine} provides a procedure tightening intervals in $\mathcal{CR}$ by $P$.
For each paradigm $P'$ involved in $\mathcal{CR}$, the exact value of $size(P,P')$ is calculated (Lines 2-3).
With these exact values, other paradigm-pairs' sizes are updated using Eq. \ref{eq_para_tria2} and \ref{eq_para_tria3}, if they make the interval tighter (Lines 5-6).

By pruning, number of merging times in Algorithms \ref{alg_pruning_merge} is no more than that in Algorithm \ref{alg_pairwise_merge}.
Meanwhile, at least one merging operation is performed in lines 7-10.
We can conclude that Algorithm \ref{alg_pruning_merge} always terminates.

\begin{algorithm}[h]
    \caption{Refine($\mathcal{CR}$, $P$)}
    \begin{algorithmic}[1]
    \label{alg_refine}
    \STATE {$paradigms\leftarrow\mathcal{P(CR)}-\{P\}$.}
    \FOR {each paradigm $P'\in paradigms$}
    \STATE {Set $lb(P,P')$ and $ub(P,P')$ with $size(P\uplus P')$.}
    \ENDFOR
    \FOR {each pair of paradigms $P_1,P_2\in paradigms$}
    \STATE {Update $lb(P_1,P_2)$ and $ub(P_1,P_2)$ using Eq. \ref{eq_para_tria3} and \ref{eq_para_tria2} according to $size(P\uplus P_1)$ and $size(P\uplus P_1)$}
    \ENDFOR
    \end{algorithmic}
\end{algorithm}
\section{Finding Paradigm-Dependency}
\label{finding}
Given a paradigm $P$, $card(P)$ can be very large, making it hard to express.
We consider of compacting $P$ by deduplicating in columns, denoted by $Comp(P)$.
For instance, in Table \ref{tab1} of Example \ref{exmp1}, \emph{Type} of the first three tuples are aligned into a new column \emph{Aligned Type} (a paradigm $P$).
By duplication, we get: $Comp(P)=\{ST\}[L]\{45\}\{126\}0[i]$.
Actually it's of the form of star-free \emph{Regular Expressions}: exactly one of the elements in braces '$\{\}$' is supposed to appear, and no more than one of elements in square brackets '$[]$' is supposed.
For instance, \emph{TL510i} is a string matches $Comp(P)$.

By now, with paradigms generated in the previous section, the notion of \emph{Paradigm Dependencies} can be defined.

\begin{definition}[Paradigm Dependency]
\label{paradigm_dependency}
a paradigm dependency $\varphi$ is of the form $(P,i)\rightarrow A$.
Where $P$ is a paradigm generated on the a string-type attribute $STR$, $i$ is a integer between $1$ and $length(P)$, $att$ and $STR$ are attributes in $\mathcal{A}$, which is the attribute set on relation $R$.

The semantic of $\varphi$ is that, given an instance $D$ on $R$, for any two tuples $t_1$ and $t_2$ in $D$ with their $STR$ values matching $Comp(P)$, and $c_1$ and $c_2$ are the two characters aligned to the $i$th column of $P$ in $t_1[STR]$ and $t_2[STR]$ respectively, if $c_1\ne c_2 \vee t_1[A]=t_2[A]$ always holds, we say that $D$ satisfies $\varphi$, denoted as $D\vdash\varphi$.
\end{definition}

The intuitive meaning of $\varphi$ is that, characters in $STR$ aligned to the specified location $i$ can decide the values of $A$.

Before defining the discovery problem of paradigm dependencies, several measures are necessary for finding high quality dependencies.

\textbf{Support} is defined as the maximum number of tuples satisfying the dependency:
\begin{equation*}
Support(D,\varphi)=\max\limits_{D'\subseteq D, D'\vdash\varphi}{|D|}
\end{equation*}
$D$ is the dataset from which dependencies are discovered.
Support is a frequency measure based on the idea that values which occur together frequently have more evidence to substantiate that they are correlated and hence are more interesting.
A threshold $support_{min}$ should be specified and every dependency discovered on $D$ should have a no lower support.

\textbf{Confidence} of $\varphi$ over $D$ can be stated as:
\begin{equation*}
Confidence(D,\varphi)=\frac{Support(D,\varphi)}{|D|}
\end{equation*}
it measures the level to which $D$ satisfies $\varphi$.
It is employed in this paper for two reasons: 1) $D$ may contains incorrect values and 2) characters with the same meaning are not always aligned to the same column.
Again, a threshold $confidence_{min}$ should be specified to filter the discovered dependencies.

It is possible that domain of the right hand attribute $A$ is very small.
For a real world instance, there are only two mainstream cpu brands (\emph{AMD} and \emph{Intel}) for computers.
Thus In $D$, it becomes possible that only a single value occurs on the right hand side attribute $A$.
In this situation, all dependencies with $A$ on the right hand side are satisfied by $D$, which have no meaningful information and should be discarded.
To handle this, the \textbf{Diversity} measure is introduced:
\begin{equation*}
Diversity(D,\varphi)=\sigma(D[A])
\end{equation*}
where $D(A)$ is the set of values of attribute $A$ occurred in $D$, and by $\sigma(\cdot)$, it counts the number of distinct values.

On the other hand, it is also possible that there are so many different characters occurring on the $i$th column in the aligned strings, that any pair of tuples have different values on the left hand side.
That means, $c_1\ne c_2$ in Definition \ref{paradigm_dependency} always hold, as a sequence, all dependency with the $i$the column on the left hand side are always satisfied.
They are meaningless and should be discarded, by introducing the measure \textbf{InnerSupport}:
\begin{equation*}
InnerSupport(D,\varphi)=\max\limits_{D'\subseteq D, D'\vdash\varphi}{\eta(P[i])}
\end{equation*}
where $\varphi$ is defined on $P$ and $i$, and $\eta(\cdot)$ denotes the number of the most frequent value's occurrence.
For example, $\eta({a,a,b,a,b,a,c,d})=4$ because $a$ is most frequent value and it occurs for $4$ times.
Similarly, the corresponding threshold $InnerSup_{min}$ is required.

Now we are ready to define the discovery problem formally.
\begin{definition}[Discovery of Paradigm Dependencies]
Given a dataset $D$ on attribute set $\mathcal{A}$, $P$ is a paradigm aligned on $D[STR]$, where $STR$ is a string-type attribute in $\mathcal{A}$,
as well as four measure thresholds: $Support_{min}$, $Confidence_{min}$, $Diversity_{min}$ and $InnerSup_{min}$.
Find out all paradigm dependencies $\varphi:(P,i)\rightarrow A$ satisfying the thresholds, where $i\in\{1,2,...,length(P)\}$ and $A\in\mathcal{A}$.
\end{definition}

Each tuple $t_l\in D$ claims a key-value pair for dependency $\varphi$, denoted by $\langle k_l, v_l\rangle$, thus a multiple set of key-value pairs can be obtained: $Claims=\{<k_l,v_l>|t_l\in D\}$.
It is not hard to obtain the support measure from $Claims$:
\begin{equation*}
Support(D,\varphi)=\sum\limits_{k}{\max\limits_{\langle k,v\rangle\in Claims}{\mathcal{N}(\langle k,v\rangle, Claims)}}
\end{equation*}
where $\mathcal{N}(\langle k,v\rangle, Claims)$ denotes the number of occurrences of $\langle k,v\rangle$ in $Claims$.

Then the confidence and diversity measures can be directly evaluated by their definitions.

Similar to support, inner support can be calculated by:
\begin{equation*}
Support(D,\varphi)=\max\limits_{\langle k,v\rangle\in Claims}{\mathcal{N}(\langle k,v\rangle, Claims)}
\end{equation*}
Actually, support and inner support are maximized meanwhile, i.e., on the same subset $D'\subseteq D$.

A straightforward method is evaluating the measurements on each paradigm $P$ generated in Algorithm \ref{alg_pruning_merge} (Line 11), then discard unqualified ones.

There are two prune strategies to reduce useless work: 1) a paradigm $P$ with $card(P)\le support_{min}$ can be discarded directly, because it is obvious that the support measure cannot be higher than $card(P)$
and 2) if both $\varphi_1:(P_1,i_1)\rightarrow A$ and $\varphi_2:(P_2,i_2)\rightarrow A$ are discarded due to low confidence, then $\varphi:((P_1\uplus P_2), i)\rightarrow A$ has a low confidence measure as well, where the $i_1$th column of $P_1$ and the $i_2$th column of $P_2$ are aligned into the $i$th column of $P$.
It is not hard to prove and we omit the detailed discussion here due to space limitation.

\section{experiments}
\label{experiments}
\subsection{Settings}
\noindent

All experiments were conducted on a machine with 256 2.2GHz Intel cpus (among which, only one cpu is used) and 3TB of RAM.
All algorithms had been implemented in Java with heap size set to 128GB.
The underlying operating system is CentOS.
\subsubsection*{Real Datasets}
We used three real datasets, namely the Notebook, CPU and RAM datasets.
All of these datasets were manually collected from the corresponding official websites.
Each record contains a string-type attribute, which is the product type, id or serial number, and can be seen as an identifier of that product.
We call this attribute \emph{ID}, values on which are to be aligned.
Other attributes are used to describe specification of a product.
For example, in the Notebook dataset, there are attributes such as \emph{Screen Size, Model of Video Card, CPU's frequency, etc.}
Not all attribute values for a product can be obtained, and when finding paradigm dependencies, key-value pairs containing $null$s are simply discarded.
Statistical information of the three subsets is listed in Table \ref{realdataset}.

\begin{table}\centering
\caption{Real world datasets statistics}
\vskip 3mm
\begin{tabular}{l|l|l|l} \hline
\label{realdataset}
Data Sets & Notebooks & CPUs & RAMs\\ \hline
\#Brands & 8 & 2 & 26 \\
\#Records & 2004 & 1592 & 1532 \\
\#Attributes & 35 & 25 & 12 \\
AVG ID length & 21.5 & 15.3 & 12.6\\
MAX ID length & 50 & 23 & 25 \\ \hline
\end{tabular}
\end{table}
\subsubsection{Synthetic Datasets}
Synthetic data were used in the efficiency evaluation for the aligning problem, so only those values on ID were generated.
When synthesizing datasets, several parameters were used, namely \emph{length of ID ($l$)}, \emph{Number of Records ($N$)}, \emph{Number of Clusters ($CN$)} and \emph{Variation Ratio ($\sigma$)}.
The charset for generating strings is $Charset=\{0,1,...,9, a,b,...z,-,\_,/,...\}$, i.e., numeric digits, letters and several other visible characters in ASCII.

Firstly, a single string with length $l$ was generated for each cluster, with characters randomly selected from $Charset$.
This string was used as a \emph{seed} to generate more similar strings for the corresponding cluster.

Nextly, $N$ strings were generated according to seeds generated before.
For each string, we select a seed at random, with uniform probabilities.
Then by copying the seed with variation, a new string can be generated similar to that seed.
By variation, it means that when copying, each character can be changed with a probability $\sigma$.
For similarity, a character is more probable to be changed into one with the same type.
For example, if the character 'f' is changed, it will change into 'd' more possible than into '0' or '\_'.
To simulate real world situations, a character is also possible to be deleted from the string.
By doing this, similar strings with different length can be obtained.
When a new string is generated, it can also be used as a seed, to avoid the situation that all strings in a cluster generated from the same single seed.

Table \ref{default_params} shows some parameters considered in the experiments.
When not explicitly stated, we use the default configuration value (highlighted in bold).

\begin{table}[t]
\centering
\caption{Experiment parameter configuration ranges.}
\label{default_params}
\begin{tabular}{|c|c|}
    \hline
    Parameter & Range \\ \hline
    String Length ($l$) & [10,...,\textbf{20}] \\
    Number of Strings ($N$) & [1000,...,\textbf{5000}] \\
    Number of Clusters ($CN$) & [10,...,\textbf{50},...,100] \\
    Variation Ratio ($\sigma$) & [0.01, ...,\textbf{0.05},0.10] \\ \hline
\end{tabular}
\end{table}

\subsubsection{Algorithms}
The Single-Merge based algorithm (Algorithm \ref{alg_single_merge}) does not fit our framework in aligning and clustering simultaneously, thus it was not evaluated experimentally.
The Pairwise merging algorithm (Algorithm \ref{alg_pairwise_merge}) provides the basic idea of this paper, it was considered as a \textbf{baseline} algorithm to evaluate the pruning techniques' efficiency.
Two pruning strategies were proposed on the baseline algorithm, namely bound based and independency based techniques.
The independency based one is very efficiency, such that when abandoned, the algorithm becomes extremely inefficient.
This made it hard to be evaluated in an acceptable time period, so we use this technique in default in all versions of the pruning based method.

As discussed before, all inequalities for bound based pruning are used either in the refining phase (i.e., Eq. \ref{eq_para_tria3} and \ref{eq_para_tria2}), or when new paradigm is created (i.e., Eq. \ref{eq_par_mono1}-\ref{eq_para_tria1}).
The former ones are basic of the pruning based method, so they were used by default.

We evaluated two versions of the pruning based algorithm: has or hasn't used Eq. \ref{eq_par_mono1}-\ref{eq_para_tria1}, the former is denoted by \textbf{Prunning+} and the latter by \textbf{Pruning-}.

\subsection{Discovered Rules}
\begin{table*}[t]
\centering
\begin{tabular}{|l|l|}
\hline
Compacted Paradigm and the location & Attribute\\
\hline
CORE I5-\{2-46-8\}\{3-6\}\underline{\textbf{\{0-79\}}}0 & Video Card\\
PHENOM X\{34\} \{89\}\underline{\textbf{\{1-9\}}}\{05\}0[E] & Frequency\\
ATHLONII \{XN\}\underline{\textbf{\{2-46E\}}}[O] \{1246-8BK\}\{0-8\}\{0-2458\}[5][TEKUX] & Core Numer\\
\hline
OCZ3\{PFGX\}\{12\}\{036\}\{03\}\{03\}\{FL\}\{BV\}\underline{\textbf{\{2-468\}}}G[K] & Memory Size\\
HX4\{23\}\underline{\textbf{\{01468\}}}C1\{2-6\}\{PF\}BK\{24\}/\{13\}\{26\} & Memory Frequency\\
P\{DVG\}\{CTV\}\{23\}\{2-468\}G\{12689\}\underline{\textbf{\{03-6\}}}\{03\}\{03\}\{FL\}LK & Transmission Standard\\
\hline
\{AEGSTU\}[H]\{1579\}\underline{\textbf{\{2-57\}}}\{5-7X\} & Screen Size\\
\{AEGSTU\}[TH]\{1579\}\{2-57PV\}\underline{\textbf{\{5-7RX\}}}[AO] & Video Card Model\\
VOSTRO-15-\{35\}56\{28\}--LAPTOP VOSTRO15-\{35\}56\{28\}-D\{12\}\underline{\textbf{\{1-35-8\}}}\{24\}5\{BRSL\} & Turbo Boost Frequency\\
\hline
\end{tabular}
\caption{Example dependencies found from the three real datasets}
\label{discovered_rules}
\end{table*}
In discovering paradigm dependencies from the three dataset, thresholds were set as: $Support_{min}=10$, $Confidence_{min} = 0.9$, $Diversity_{min} = 5$ and $InnerSup_{min} = 5$.
As to the distance function $d$, we specified zero distance for the same characters, and $0.5$ for different ones but with the same type, $1.5$ for ones with different type.
For example, $d('a', 'a')=0$, $d('a', 'b')=0.5$ and $d('a', '9')=1.5$.

Eventually, there were $33$, $91$ and $8$ dependencies discovered from the Notebook, RAM and CPU dataset respectively.
All these rules are verified true manually.
With lower thresholds, the result set can grows much larger.

In Table \ref{discovered_rules}, several discovered dependencies are listed in three groups, with each group corresponding to a single dataset.
The paradigm in the dependency is demonstrated in a compact format, actually it is a star-free Regular Expression: exactly one of the elements in braces '$\{\}$' is supposed to appear, and no more than one of elements in square brackets '$[]$' is supposed, delimiters between optional elements are omitted to save space.
The locations are marked by underlines for intuition.

Take the first dependency for example, for the CPU model \emph{I5}, it means that the second element from backward can be a numeric digit from $0$ to $9$ except $8$, and its value tells the model of video card integrated.

%\begin{table*}[!hbp]
%\centering
%\begin{tabular}{|l|c|c|c|c|c|}
%\hline
%Paradigm & Attribute & Support & Confidence & Diversity & Inner Sup\\
%\hline
%CORE I5-\{2-46-8\}\{3-6\}\underline{\{0-79\}}0 & Video Card & 20 & 0.91 & 5 & 7 \\
%PHENOM X\{34\}\_\{89\}\underline{\{1-9\}}\{05\}0[E] & CPU Frequency & 24 & 1.0 & 9 & 6 \\
%ATHLONII\_\{XN\}\underline{\{2-46E\}}[O]\_\{1246-8BK\}\{0-8\}\{0-2458\}[5][TEKUX] & Core Numer & 56 & 1.0 & 5 & 23 \\
%OCZ3\{PFGX\}\{12\}\{036\}\{03\}\{03\}\{FL\}\{BV\}\underline{\{2-468\}}G[K] & Memory Size & 17 & 1.0 & 5 & 7 \\
%HX4\{23\}\underline{\{01468\}}C1\{2-6\}\{PF\}BK\{24\}/\{13\}\{26\} & Memory Frequency & 15 & 1.0 & 5 & 5 \\
%P\{DVG\}\{CTV\}\{23\}\{2-468\}G\{12689\}\underline{\{03-6\}}\{03\}\{03\}\{FL\}LK & Transmit Standard & 14 & 0.93 & 5 & 5 \\
%\{AEGSTU\}[H]\{1579\}\underline{\{2-57\}}\{5-7X\} & Screen Size & 20 & 1.0 & 5 & 7 \\
%\{AEGSTU\}[TH]\{1579\}\{2-57PV\}\underline{\{5-7RX\}}[AO] & Video Card Model & 20 & 0.91 & 5 & 8 \\
%VOSTRO-15-\{35\}56\{28\}--LAPTOP VOSTRO15-\{35\}56\{28\}-D\{12\}\underline{\{1-35-8\}}\{24\}5\{BRSL\} & CPU Highest Frequency & 16 & 0.94 & 5 & 5 \\
%... & ... & ... & ... & ... & ... \\
%\hline
%\end{tabular}
%\caption{Several notebook types of \emph{Thinkpad}}
%\label{tab1}
%\end{table*}

\subsection{Efficiency}
\noindent{\textbf{On Real Datasets}}
\begin{figure}[h]
  \centering
  \includegraphics[width=0.3\textwidth]{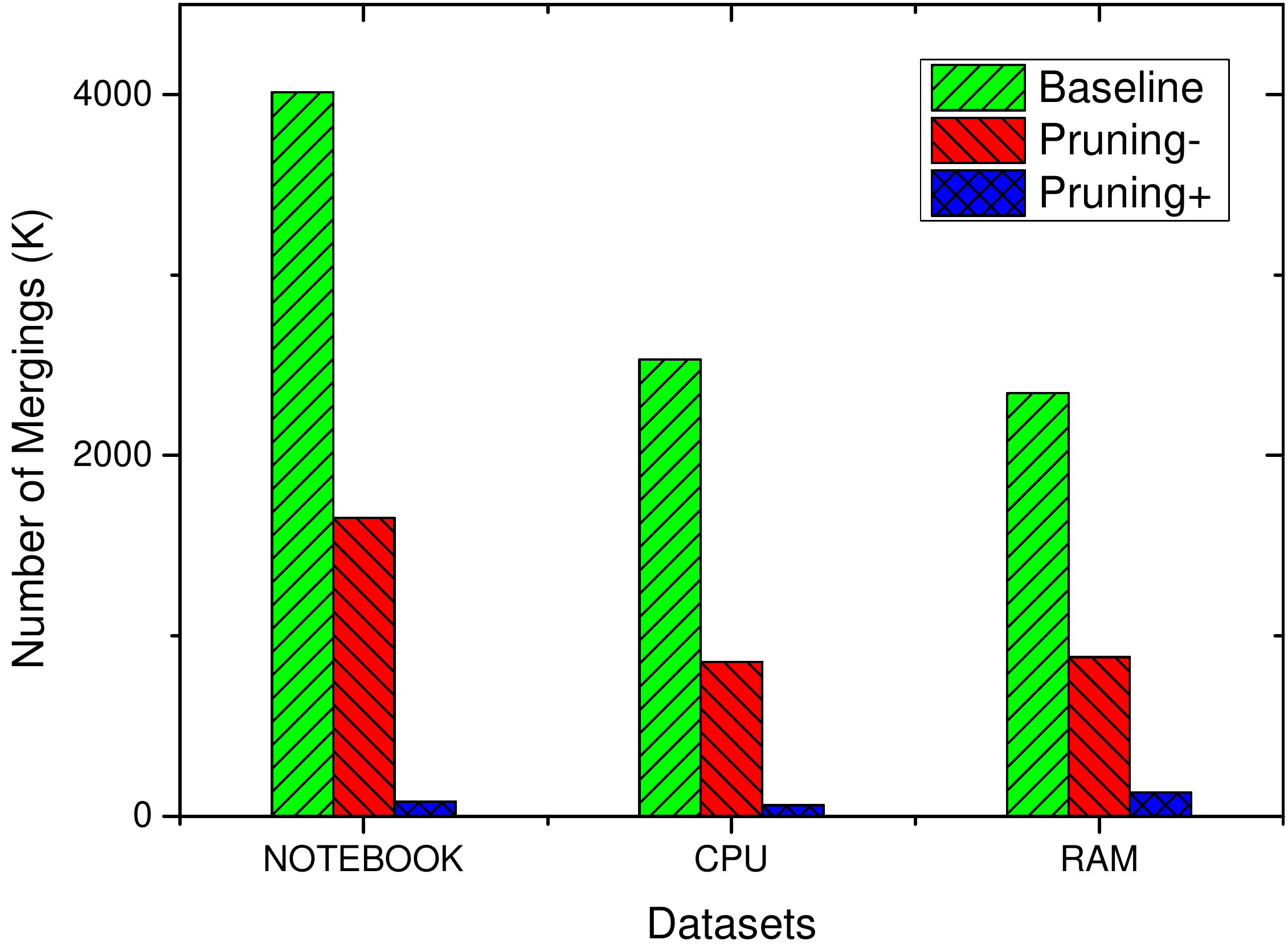}
  \caption{Number of Merging Operations of Different Methods.}
  \label{real_merge}
\end{figure}

We concern the running time, number of merging operations and number of refining operations in different methods.
From Fig. \ref{real_merge} we can see that Baseline is of the highest number of merging, Pruning- reduced that by about two-thirds, and Pruning+ reduced it greatly.

However, in Fig. \ref{real_time}, the time consumption shows a quite different result: Pruning+ remained the most efficient, and Pruning- becomes very inefficient.
That is because although some merging operations are saved, even much more time was wasted in the refining phase, maintaining necessary data structures becomes very costly if refining operation is too frequent.
Fig. \ref{real_refine} shows comparison of the number of refining between Pruning+ and Pruning-.
\begin{figure}[h]
  \centering
  \includegraphics[width=0.3\textwidth]{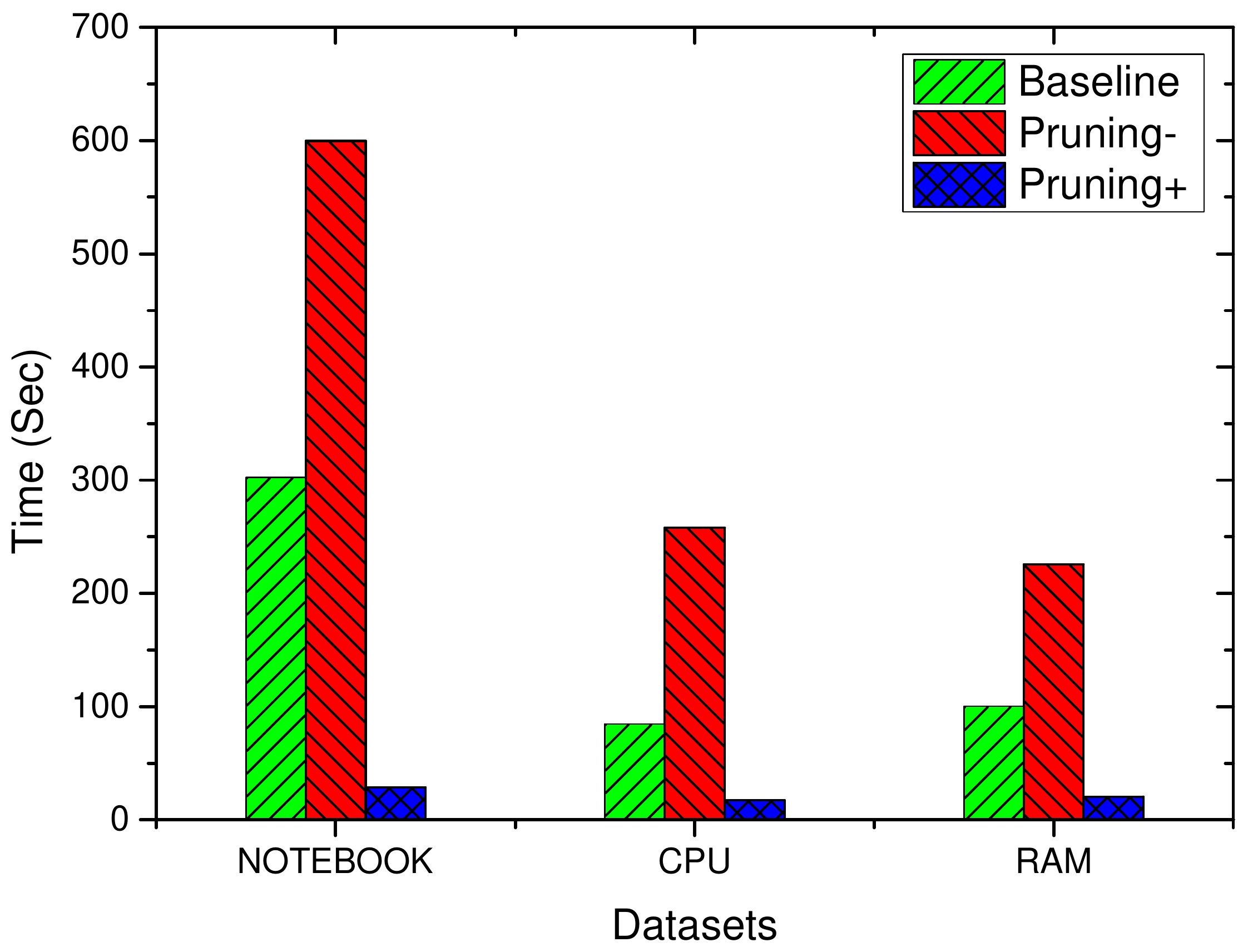}
  \caption{Time Consuming of Different Methods.}
  \label{real_time}
\end{figure}

\begin{figure}[h]
  \centering
  \includegraphics[width=0.3\textwidth]{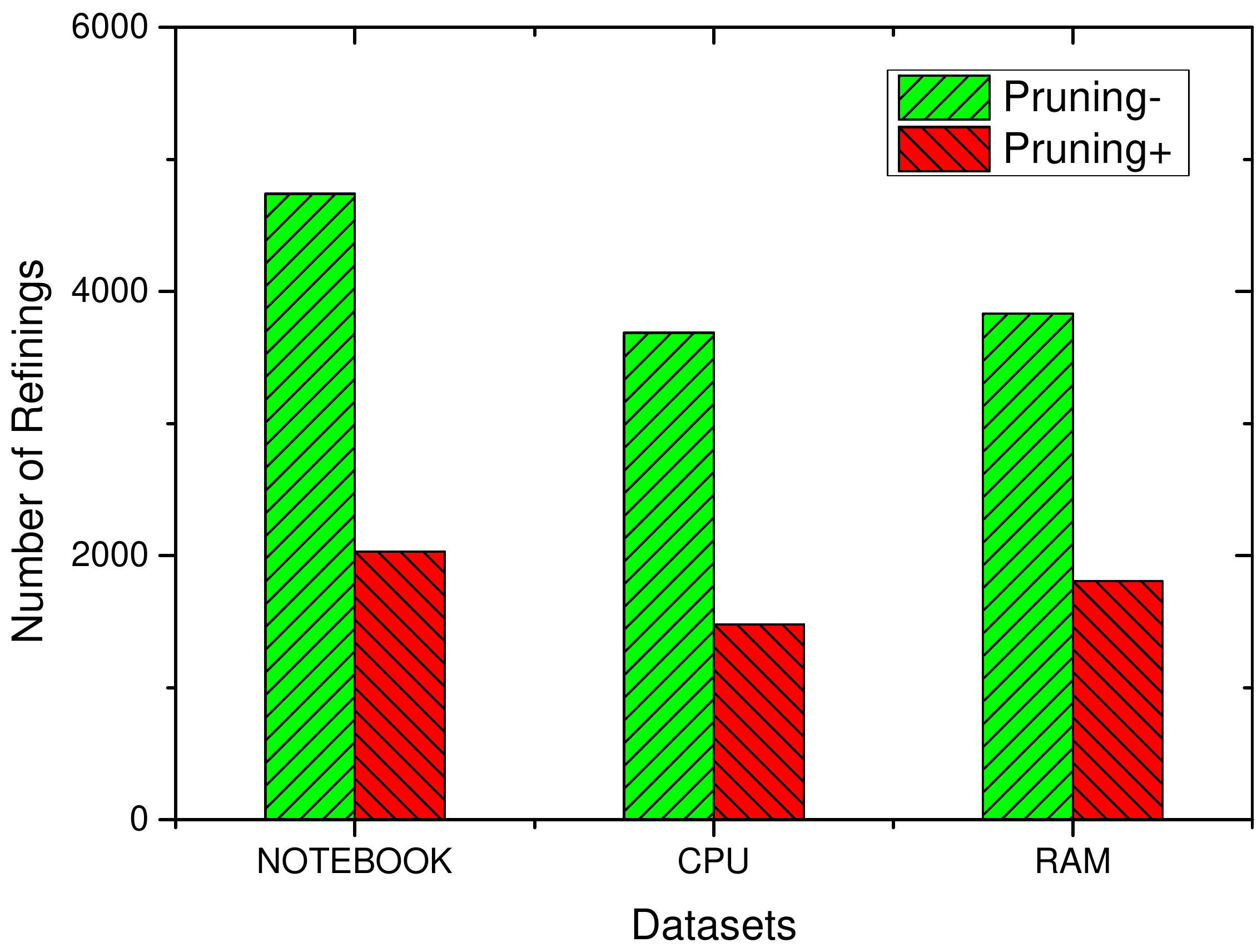}
  \caption{Number of Refine Operations Comparison.}
  \label{real_refine}
\end{figure}

In Fig. \ref{real_inner_refine} are statistics of the number of refines in selecting the best pair of paradigms.
It shows that on all of the three datasets, in most of the time, Pruning+ refined the critical set $CR$ only once (in Line 9 of Algorithm \ref{alg_pruning_merge}).
Even in the worst case, only six refining operations were required.
\begin{figure}[h]
  \centering
  \includegraphics[width=0.3\textwidth]{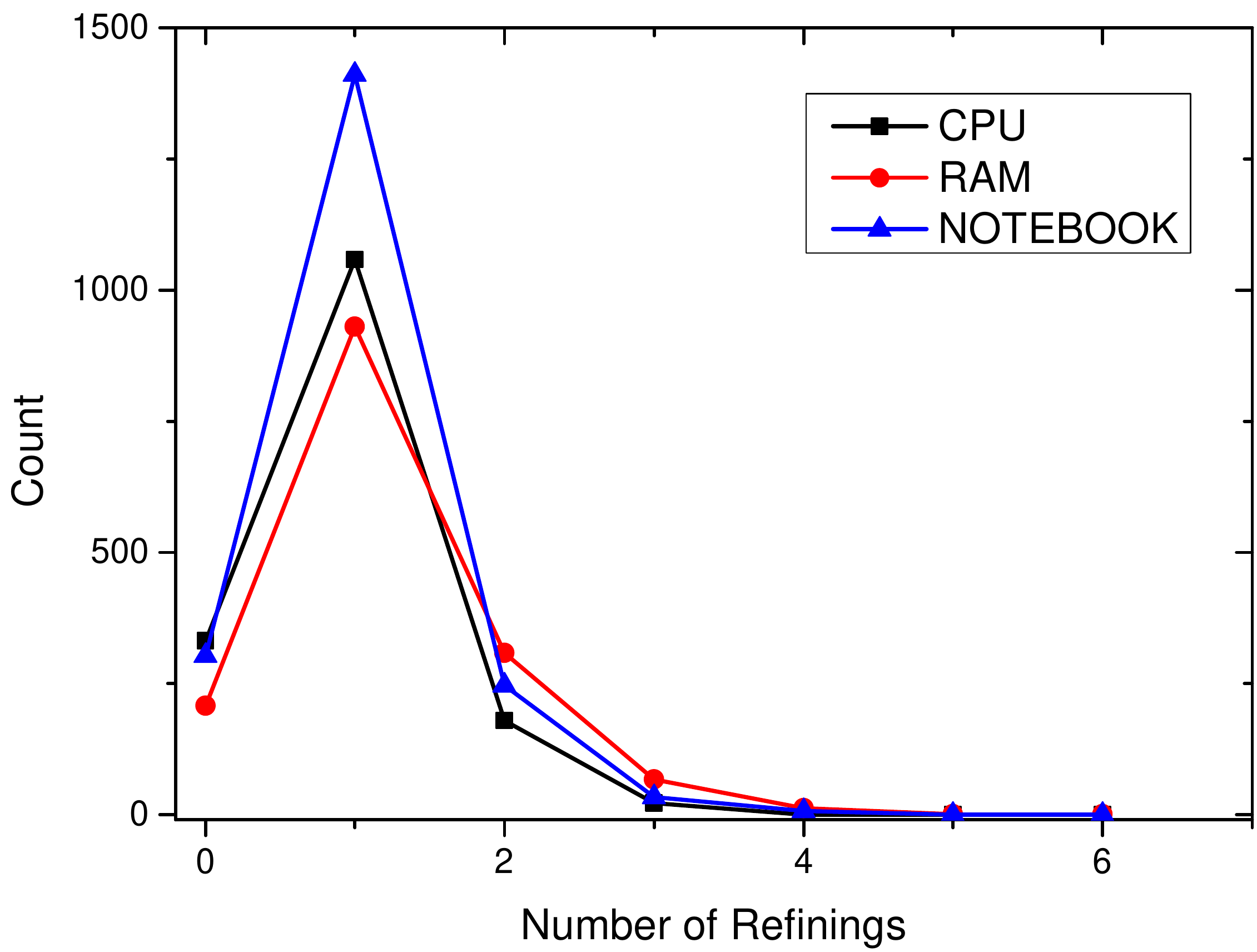}
  \caption{Statistics of Refine Counts in Pruning+.}
  \label{real_inner_refine}
\end{figure}

\noindent{\textbf{On Synthetical Datasets}}
In Fig. \ref{syn_vn_merge}, the number of Baseline's merging operations remains a constant, because it is unrelated to strings' length.
Pruning based methods can reduce the merging numbers to a great extent, especially Pruning+.

\begin{figure}[h]
  \centering
  \includegraphics[width=0.3\textwidth]{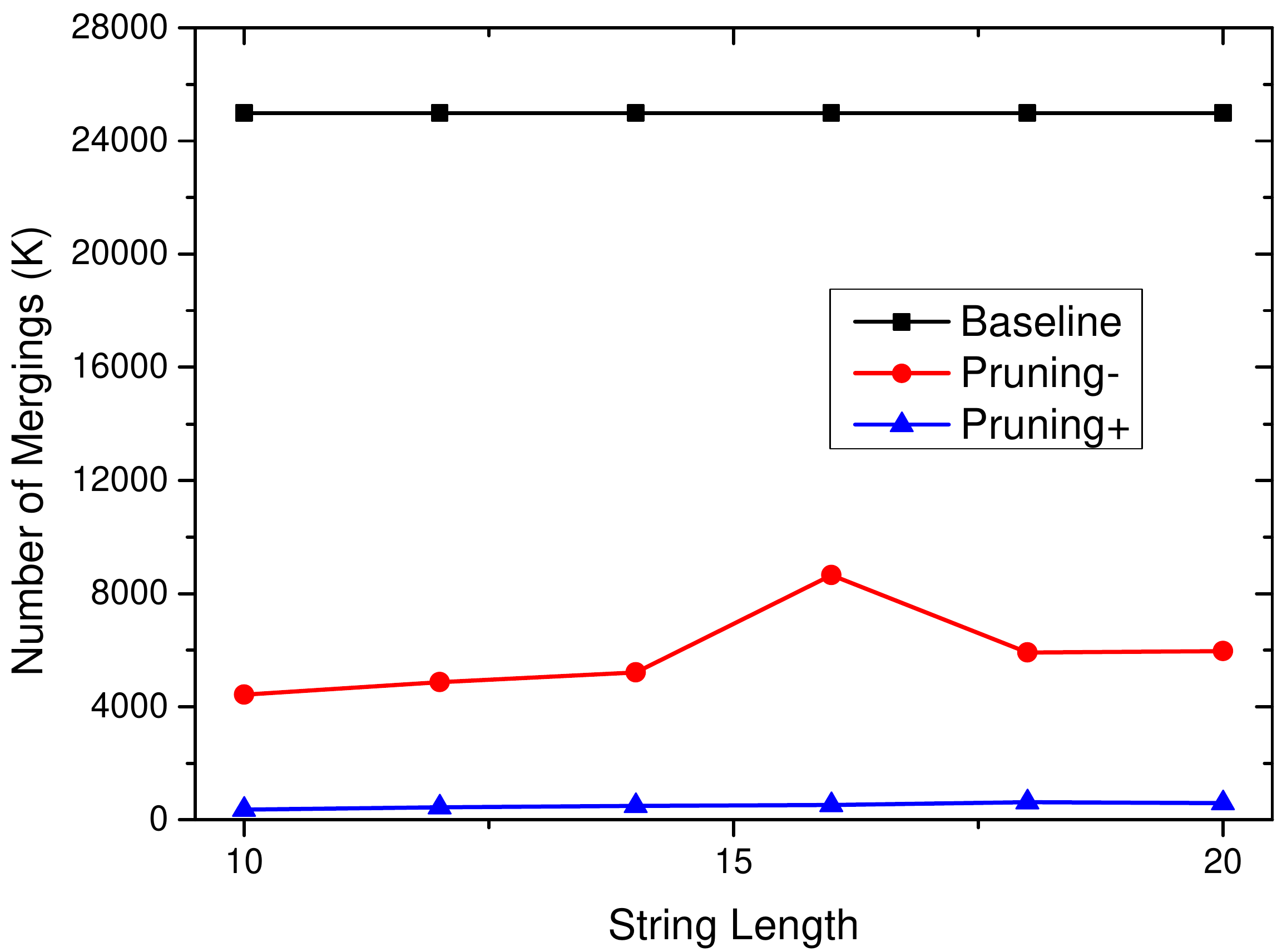}
  \caption{Number of Merging Operations When Varying String Length.}
  \label{syn_vn_merge}
\end{figure}

As it does on real datasets, Pruning- has the worst efficiency, as shown in Fig. \ref{syn_vn_time}.
\begin{figure}[h]
  \centering
  \includegraphics[width=0.3\textwidth]{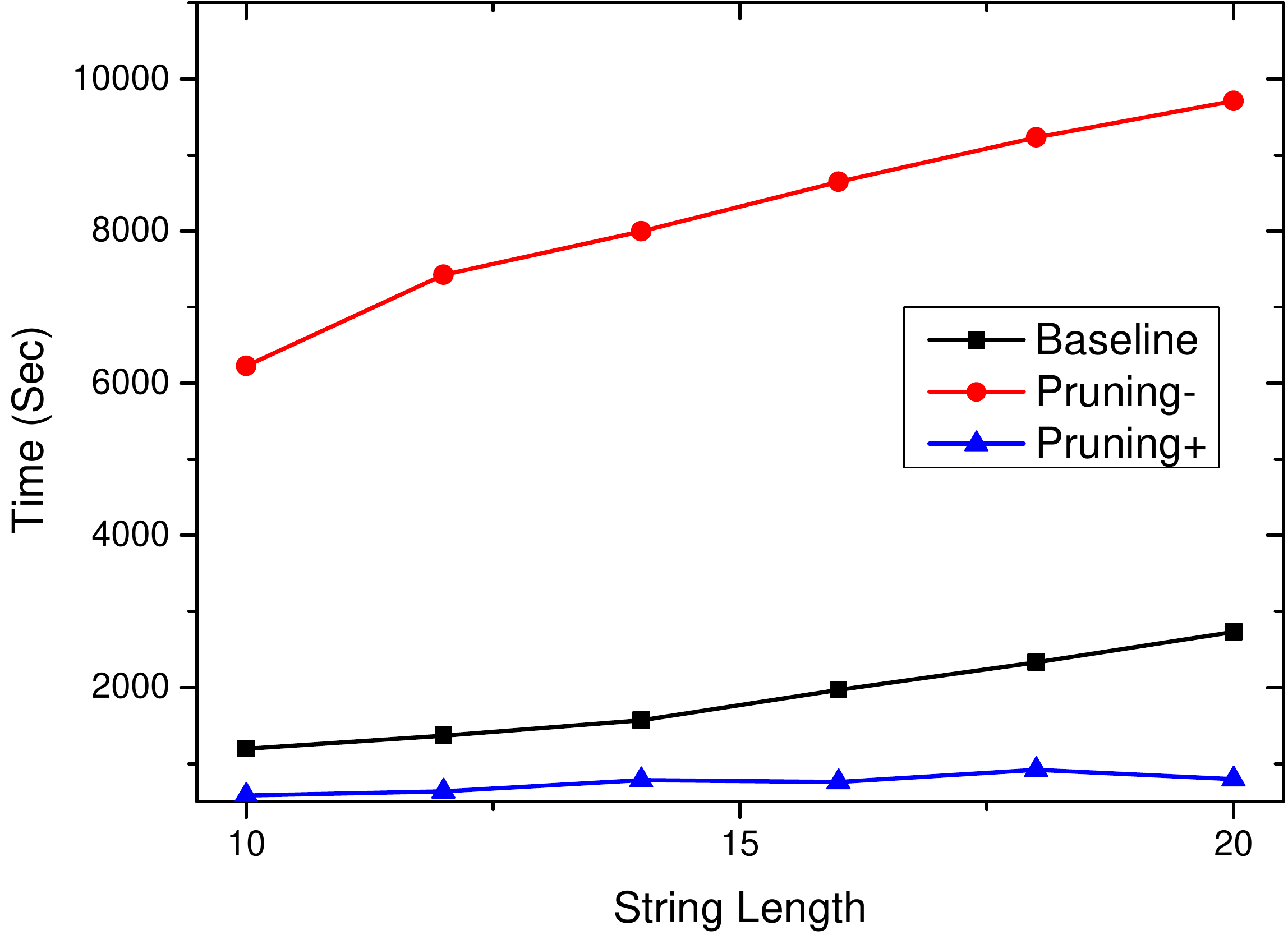}
  \caption{Time Consuming When Varying String Length.}
  \label{syn_vn_time}
\end{figure}

Fig. \ref{syn_vNN_merge} and \ref{syn_vNN_time} illustrate the number of mergings and time consumed in different algorithms, when varying the number of input strings ($N$).
It can be observed that Baseline's merging number is proportional to $N^2$, and Pruning+ has almost linear time costs and merging operations, making it the best algorithm in efficiency.
\begin{figure}[h]
  \centering
  \includegraphics[width=0.3\textwidth]{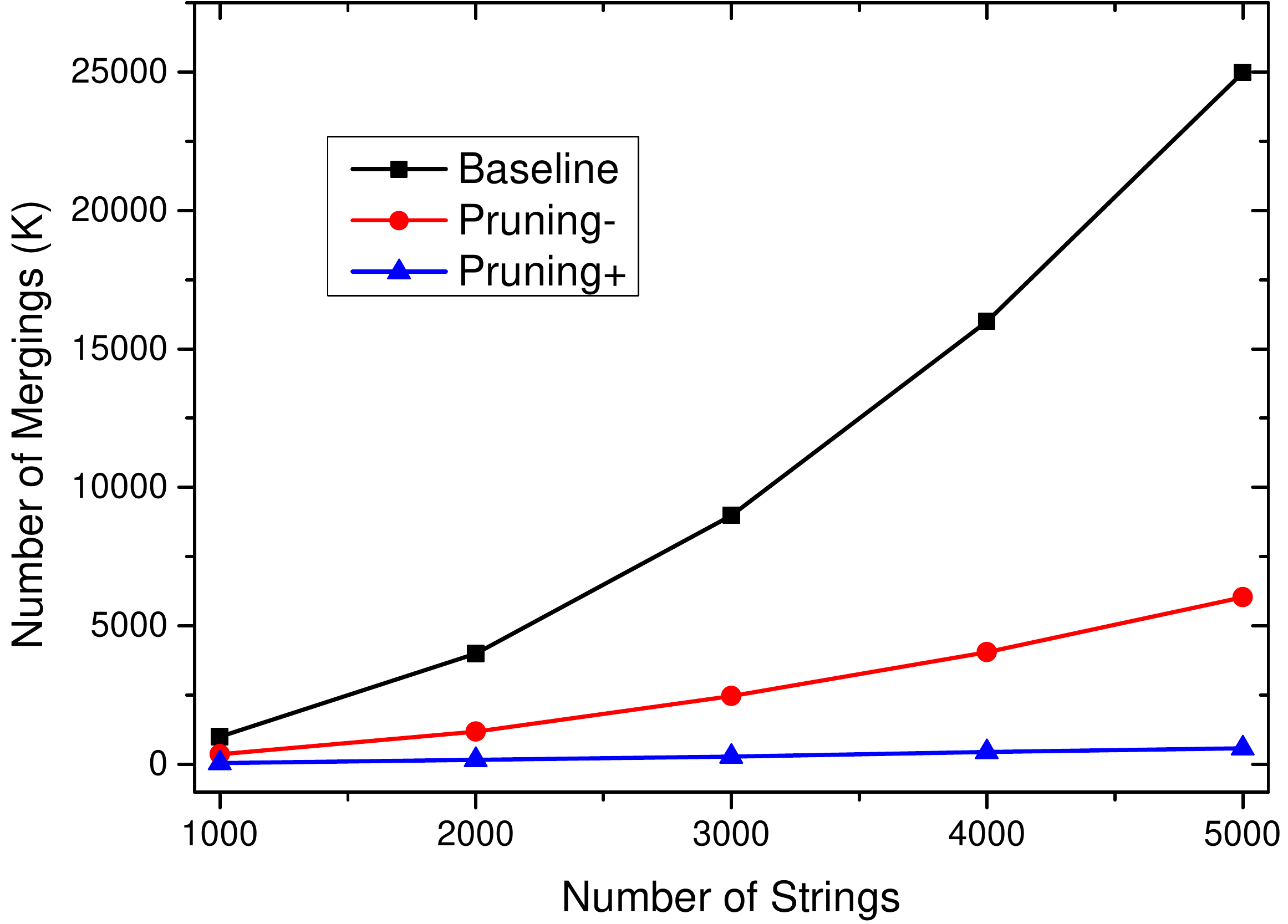}
  \caption{Number of Merging Operations When Varying Number of Strings.}
  \label{syn_vNN_merge}
\end{figure}

\begin{figure}[h]
  \centering
  \includegraphics[width=0.3\textwidth]{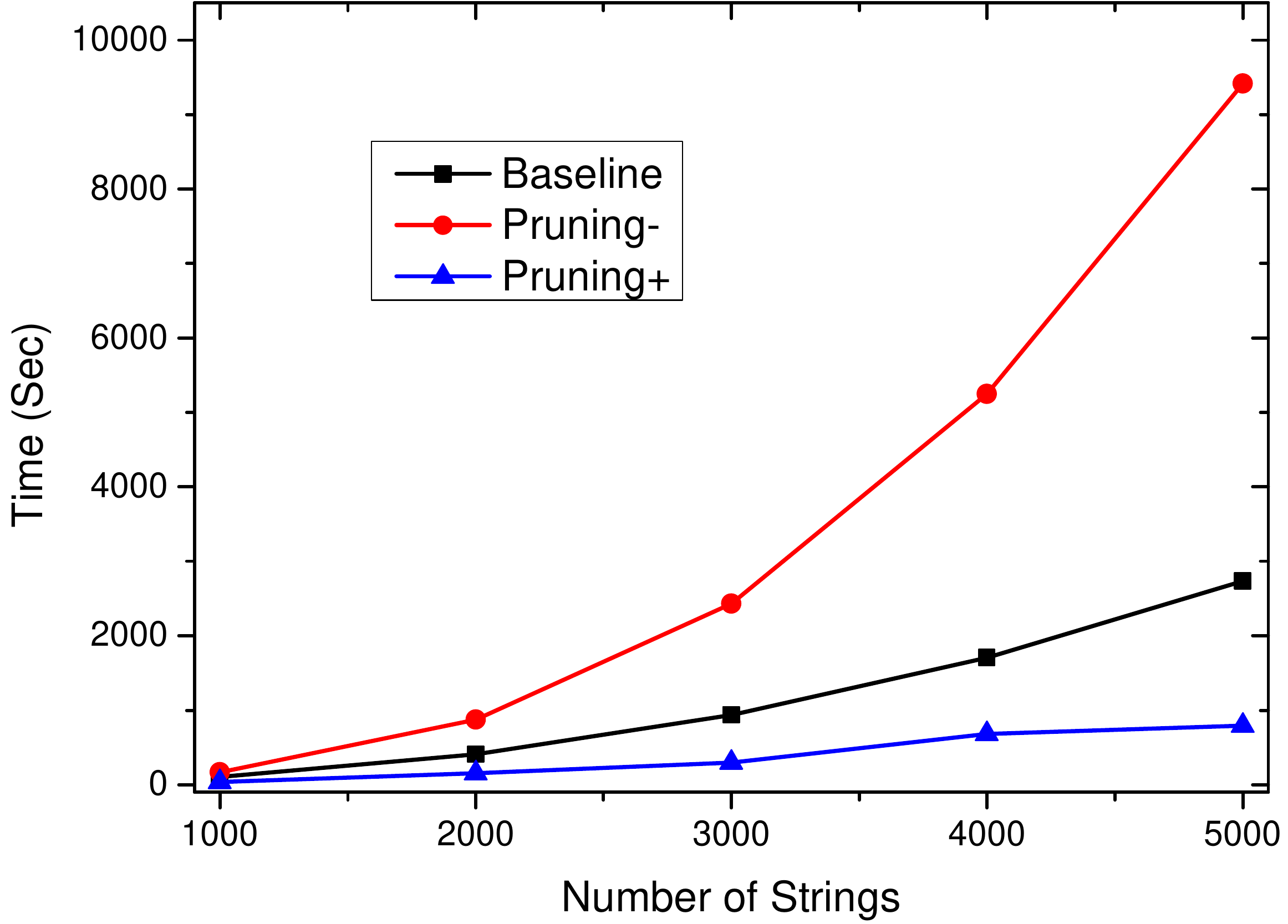}
  \caption{Time Consuming When Varying Number of Strings.}
  \label{syn_vNN_time}
\end{figure}

From Fig. \ref{syn_vcs_time}, it can be seen that by changing number of clusters, or the variation ratio $\sigma$, the running time of Baseline and Pruning+ did not change regularly.

\begin{figure}
  \centering
  \subfigure[]{
    \includegraphics[width=0.225\textwidth]{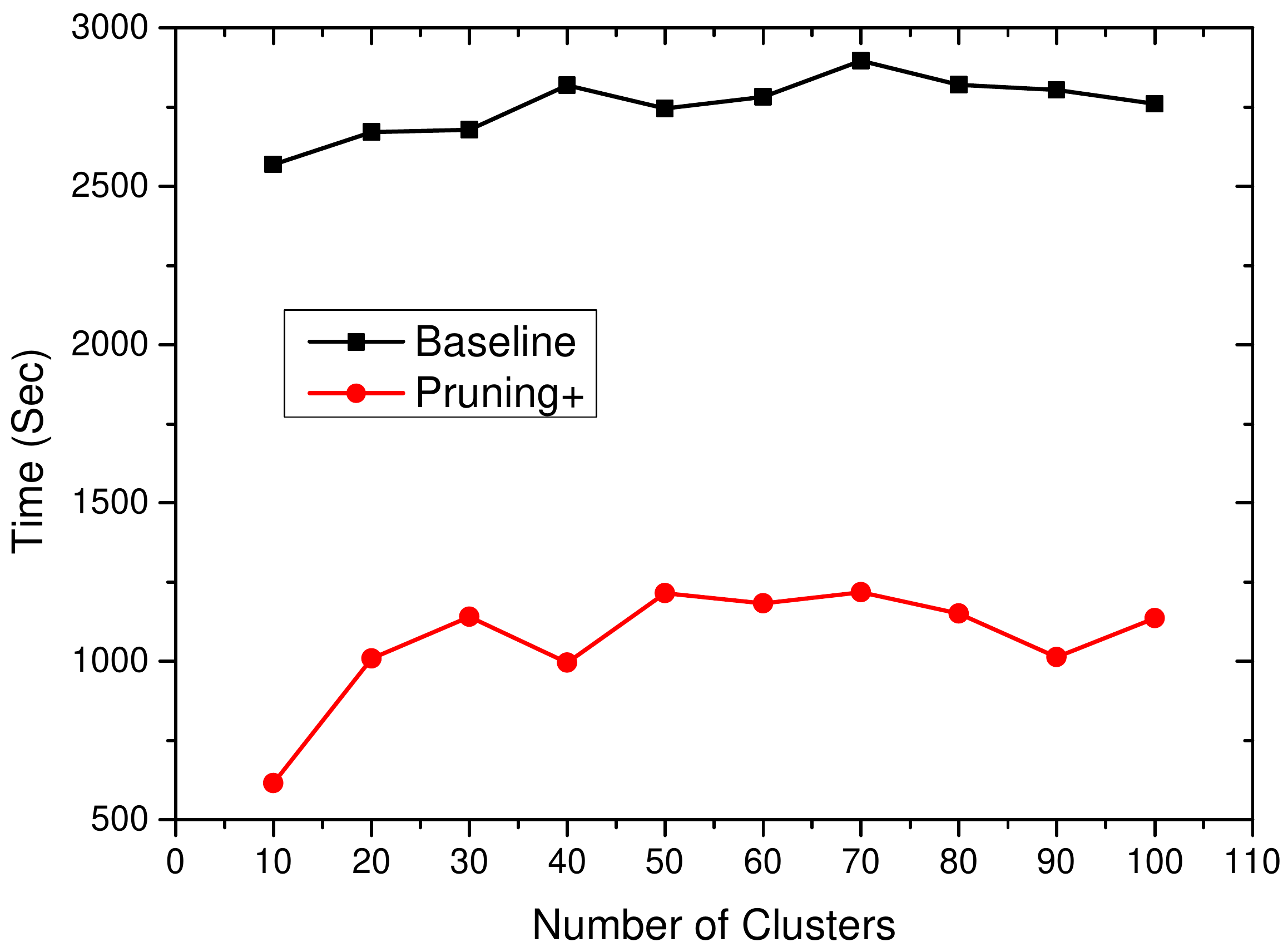}}
  \subfigure[]{
    \includegraphics[width=0.225\textwidth]{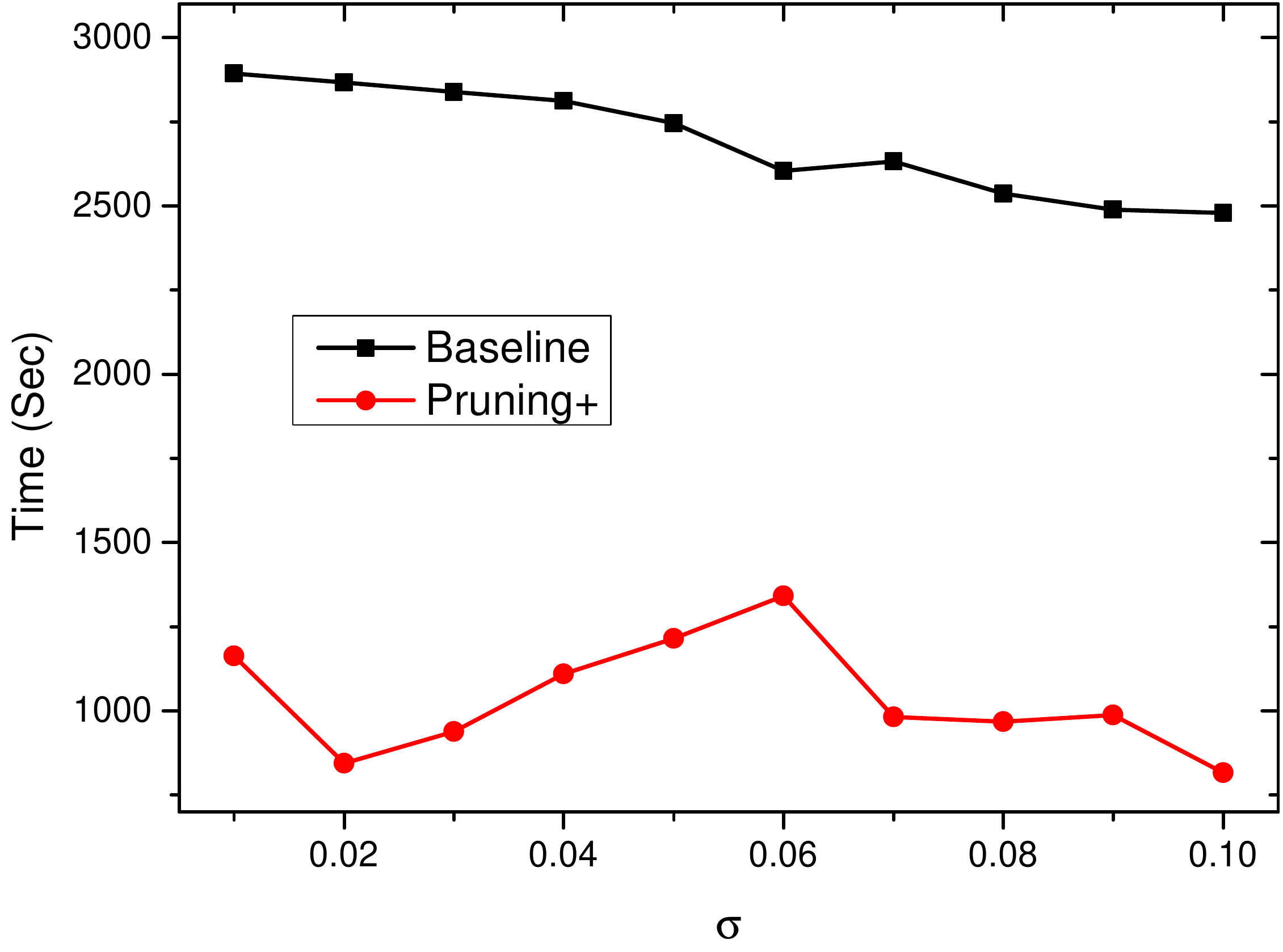}}
  \caption{An example of interval refining}
\label{syn_vcs_time}
\end{figure}

%\begin{figure}[b]
%  \centering
%  \includegraphics[width=0.3\textwidth]{syn_vc_time.pdf}
%  \caption{Time Consuming When Varying Number of Clusters.}
%  \label{syn_vc_time}
%\end{figure}

%\begin{figure}
%  \centering
%  \includegraphics[width=0.3\textwidth]{syn_vs_time.pdf}
%  \caption{Time Consuming When Varying $\sigma$.}
%  \label{syn_vs_time}
%\end{figure}

\section{Conclusion}
\label{conclusion}
In this paper, a new type of dependencies, \emph{Paradigm Dependencies} has been proposed, in which the left hand side is part a regular expression like paradigm.
To discover such dependencies, a framework has been proposed to align and cluster meaningful strings simultaneously.
The aligning problem has been proved in \emph{NP-Complete}, and a greedy algorithm was introduced in which the clustering and aligning tasks can be combined together seamlessly.
Due to the greedy algorithm's high time complexity, several pruning strategies with theoretical support were proposed to reduce the running time.
Then discovery of paradigm dependencies on the generated paradigms have been defined and discussed, based on four measurements.
Finally, our methods' effectiveness and efficiency have been verified on three real world datasets as well as synthetical datasets.

For future work, we consider of several aspects: 1) In our framework, dependencies are discovered on already merged paradigms.
Actually, correlations between a paradigm's column and an attribute's values can help aligning the strings more wisely.
So by changing our two-step framework into a iterative and interactive one may improve the effectiveness to some extent.
2) Due to the high complexity, we consider of redesigning the greedy algorithm, e.g., by parallelization, or trading off between effectiveness and efficiency, etc.
3) In the discovering phase, we assumed a single position in the left hand side.
Actually, it is possible that multiple elements together in a string indicate another attribute's value, which makes the problem more complicated.
4) The purpose of introducing paradigm dependencies is to handle dirty data, so it is nature to study data cleaning problems, such as inconsistent data repairing and missing values imputation, using these new proposed rules.

\section*{Acknowledgment}
This paper was partially supported by the Key Research and Development Plan of National Ministry of Science and Technology under grant No. 2016YFB1000703, the Key Program of the National Natural Science Foundation of China under Grant No. 61190115, 61472099, 61632010, U1509216, National Sci-Tech Support Plan 2015BAH10F01, the Scientific Research Foundation for the Returned Overseas Chinese Scholars of Heilongjiang Province LC2016026 and MOE¨CMicrosoft Key Laboratory of Natural Language Processing and Speech, Harbin Institute of Technology.

\end{document}